\documentclass[11pt]{article}
\usepackage[dvips]{graphicx, color}

\usepackage{amsmath} %[intlimits]
\usepackage{amssymb}
\usepackage{amsthm}
\usepackage{amsxtra}
\usepackage{amscd}
\usepackage{bbm}
\usepackage{mathrsfs}
\usepackage{bm}

\renewcommand{\cal}{\mathcal}

\newcommand{\ul}[1]{\underline{#1} \!\,} %underline
\newcommand{\ol}[1]{\overline{#1} \!\,} %overline
\newcommand{\ee}{\mathrm{e}} \newcommand{\me}{\mathrm{e}}
\newcommand{\ii}{\mathrm{i}} \newcommand{\mi}{\mathrm{i}}
\newcommand{\dd}{\mathrm{d}} 

\newcommand{\deq}{\mathrel{\mathop:}=}
\newcommand{\eqd}{=\mathrel{\mathop:}}
\newcommand{\id}{\mspace{2mu}\mathrm{i}\mspace{-0.6mu}\mathrm{d}} %identity map
\newcommand{\umat}{\mathbbmss{1}} %unit matrix
\renewcommand{\epsilon}{\varepsilon}
\renewcommand{\leq}{\leqslant}
\renewcommand{\geq}{\geqslant}

\newcommand{\R}{\mathbb{R}}
\newcommand{\C}{\mathbb{C}}
\newcommand{\N}{\mathbb{N}}

\newcommand{\SSS}{\mathrm{S}}

\newcommand{\HHH}{H}
\newcommand{\CCC}{C}
\newcommand{\LLL}{L}

\newcommand{\pb}[1]{\bigl({#1}\bigr)}
\newcommand{\pB}[1]{\Bigl({#1}\Bigr)}
\newcommand{\pbb}[1]{\biggl({#1}\biggr)}

\newcommand{\qb}[1]{\bigl[{#1}\bigr]}
\newcommand{\qB}[1]{\Bigl[{#1}\Bigr]}
\newcommand{\qbb}[1]{\biggl[{#1}\biggr]}

\newcommand{\h}[1]{\{{#1}\}}
\newcommand{\hb}[1]{\bigl\{{#1}\bigr\}}

\newcommand{\abs}[1]{\lvert #1 \rvert}
\newcommand{\absb}[1]{\big\lvert #1 \big\rvert}
\newcommand{\absB}[1]{\Big\lvert #1 \Big\rvert}
\newcommand{\absbb}[1]{\bigg\lvert #1 \bigg\rvert}

\newcommand{\norm}[1]{\lVert #1 \rVert}
\newcommand{\normb}[1]{\big\lVert #1 \big\rVert}
\newcommand{\normB}[1]{\Big\lVert #1 \Big\rVert}
\newcommand{\normbb}[1]{\bigg\lVert #1 \bigg\rVert}

\newcommand{\scalar}[2]{\langle{#1} \mspace{2mu}, {#2}\rangle}
\newcommand{\scalarb}[2]{\big\langle{#1} \mspace{2mu}, {#2}\big\rangle}
\newcommand{\scalarB}[2]{\Big\langle{#1} \,\mspace{2mu},\, {#2}\Big\rangle}

\newcommand{\bra}[1]{\langle #1 |}
\newcommand{\brab}[1]{\big\langle #1 \big|}

\newcommand{\ket}[1]{| #1 \rangle}
\newcommand{\ketb}[1]{\big| #1 \big\rangle}

\DeclareMathOperator*{\wstarlim}{w*-lim}

\DeclareMathOperator{\tr}{Tr}

\DeclareMathOperator*{\esssup}{ess\,sup}
\DeclareMathOperator{\ad}{ad}

\numberwithin{equation}{section}
\numberwithin{figure}{section}

\theoremstyle{plain} %plain, definition, remark
\newtheorem{theorem}{Theorem}[section]
\newtheorem*{theorem*}{Theorem}
\newtheorem{lemma}[theorem]{Lemma}
\newtheorem*{lemma*}{Lemma}
\newtheorem{corollary}[theorem]{Corollary}
\newtheorem*{corollary*}{Corollary}

\newtheorem*{proposition*}{Proposition}
\theoremstyle{definition} %plain, definition, remark

\newtheorem*{definition*}{Definition}

\newtheorem*{example*}{Example}
\newtheorem{remark}[theorem]{Remark}

\newtheorem*{remark*}{Remark}
\newtheorem*{remarks*}{Remarks}

\usepackage{cite}
\usepackage[dvips]{graphicx}
\usepackage[nottoc,notlof,notlot]{tocbibind}

\usepackage{color}
\usepackage{psfrag}

\usepackage{verbatim}

\newcommand{\md}{\dd}
\newcommand{\quant}{^{\!\widehat{\;\;\;\;}}_N}
\newcommand{\A}{\mathrm{A}}

\begin{document}
\title{On the Mean-Field Limit of Bosons with Coulomb Two-Body Interaction}
\author{J\"urg Fr\"ohlich\footnote{juerg@itp.phys.ethz.ch} \qquad Antti Knowles\footnote{aknowles@itp.phys.ethz.ch} \qquad Simon Schwarz\footnote{sschwarz@itp.phys.ethz.ch} 
\\
\\
Institute of Theoretical Physics
\\
ETH H\"onggerberg,
\\
CH-8093 Z\"urich, 
\\
Switzerland.
}
\maketitle

\begin{abstract}
In the mean-field limit the dynamics of a quantum Bose gas is described by a Hartree equation. We present a simple method for proving the convergence of the microscopic quantum dynamics to the Hartree dynamics when the number of particles becomes large and the strength of the two-body potential tends to 0 like the inverse of the particle number. Our method is applicable for a class of singular interaction potentials including the Coulomb potential. We prove and state our main result for the Heisenberg-picture dynamics of ``observables'', thus avoiding the use of coherent states. Our formulation shows that the mean-field limit is a ``semi-classical'' limit.
\end{abstract}

\section{Introduction}
Whenever many particles interact by means of weak two-body potentials, one expects that the potential felt by any one particle is given by an average potential generated by the particle density. In this \emph{mean-field} regime, one hopes to find that the emerging dynamics is simpler and less encumbered by tedious microscopic information than the original $N$-body dynamics.

The mathematical study of such problems has quite a long history. In the context of classical mechanics, where the mean-field limit is described by the Vlasov equation, the problem was successfully studied by Braun and Hepp \cite{BraunHepp}, as well as Neunzert \cite{Neunzert}. The mean-field limit of quantum Bose gases was first addressed in the seminal paper \cite{Hepp} of Hepp. We refer to \cite{ErdosYau} for a short discussion of some subsequent results. The case with a Coulomb interaction potential was treated by Erd\H{o}s and Yau in \cite{ErdosYau}. Recently, Rodnianski and Schlein \cite{RodnianskiSchlein} have derived explicit estimates for the rate of converge to the mean-field limit, using the methods of \cite{Hepp} and \cite{GinibreVelo}. A sharper bound on the rate of convergence in the case of a sufficiently regular interaction potential was derived by Schlein and Erd\H{o}s \cite{SchleinErdos}, by using a new method inspired by Lieb-Robinson inequalities. In \cite{NarnhoferSewell, FrohlichGraffiSchwarz}, the mean-field limit ($N\to \infty$) and the classical limit were studied simultaneously. A conceptually quite novel approach to studying mean-field limits was introduced in \cite{FrohlichKnowles}. In that paper, the time evolution of quantum and corresponding ``classical'' observables is studied in the Heisenberg picture, and it is shown that ``time evolution commutes with quantization'' up to terms that tend to 0 in the mean-field (``classical'') limit, which is a \emph{Egorov-type result}.

In this paper we present a new, simpler way of handling singular interaction potentials. It yields a Egorov-type formulation of convergence to the mean-field limit, thus obviating the need to consider particular (traditionally coherent) states as initial conditions. Another, technical, advantage of our method is that it requires no regularity (traditionally $\HHH^1$- or $\HHH^2$-regularity) when applied to coherent states.

Such kinds of results were first obtained by Egorov \cite{Egorov} for the semi-classical limit of a quantum system. Roughly, the statement is that time-evolution commutes with quantization in the semi-classical limit. We sketch this in a simple example: Let us start with a classical Hamiltonian system of a finite number $f$ of degrees of freedom. The classical algebra of observables $\mathfrak{A}$ is given by (some subalgebra of) the Abelian algebra of smooth functions on the phase space $\Gamma \deq \R^{2f}$. Let $H \in \mathfrak{A}$ be a Hamilton function. Together with the symplectic structure on $\Gamma$, $H$ generates a symplectic flow $\phi^t$ on $\Gamma$. Now we define a quantization map $\widehat{(\cdot)}_\hbar: \mathfrak{A} \to \widehat{\mathfrak{A}}$, where $\widehat{\mathfrak{A}}$ is some subalgebra of $\mathcal{B}(\LLL^2(\R^f))$. For concreteness, let $\widehat{(\cdot)}_\hbar$ be Weyl quantization with deformation parameter $\hbar$. This implies that
\begin{equation*}
\qb{\widehat{A}_\hbar, \widehat{B}_\hbar} \;=\; \frac{\hbar}{\mi} \widehat{\h{A,B}}_\hbar + O(\hbar^2)\,,
\end{equation*}
for $\hbar \to 0$.
The quantized Hamilton function defines a 1-parameter group of automorphisms on $\widehat{\mathfrak{A}}$ through
\begin{equation*}
\mathbf{A} \;\mapsto\; \me^{\mi t \widehat{H}_\hbar / \hbar} \,  \mathbf{A} \, \me^{- \mi t \widehat{H}_\hbar / \hbar}\,, \qquad \mathbf{A} \in \widehat{\mathfrak{A}}\,.
\end{equation*}
A Egorov-type semi-classical result states that, for all $A \in \mathfrak{A}$ and $t \in \R$,
\begin{equation*}
\widehat{(A \circ \phi^t)}_\hbar \;=\; \me^{\mi t \widehat{H}_\hbar / \hbar} \,  \widehat{A}_\hbar \, \me^{- \mi t \widehat{H}_\hbar / \hbar} + R_\hbar(t)\,,
\end{equation*}
where $\norm{R_\hbar(t)} \to 0$ as $\hbar \to 0$.

This approach identifies the semi-classical limit as the converse of quantization. In a similar fashion, we identify the mean-field limit as the converse of ``second quantization''. In this case the deformation parameter is not $\hbar$, but $N^{-1}$, a parameter proportional to the coupling constant. We consider the mean-field dynamics (given by the Hartree equation in the case of bosons), and view it as the Hamiltonian dynamics of a classical Hamiltonian system. We show that its quantization describes $N$-body quantum mechanics, and that the ``semi-classical'' limit corresponding to $N^{-1} \to 0$ takes us back to the Hartree dynamics.

We sketch the key ideas behind our strategy.
\begin{itemize}
\item[(1)]
Use the Schwinger-Dyson expansion to construct the Heisenberg-picture dynamics of $p$-particle operators
\begin{equation*}
\ee^{\ii t H_N} \, \widehat{\A}_N(a^{(p)}) \, \ee^{- \ii t H_N}
\end{equation*}
(in the notation of Section \ref{section: setup}).
\item[(2)]
Use \emph{Kato smoothing} plus \emph{combinatorial estimates} (counting of graphs) to prove convergence of the Schwinger-Dyson expansion on $N$-particle Hilbert space, uniformly in $N$ and for small $\abs{t}$. Diagrams containing $l$ loops yield a contribution of order $N^{-l}$.
\item[(3)] Use Kato smoothing plus combinatorial estimates to prove convergence of the iterative solution of the Hartree equation, for small $\abs{t}$.
\item[(4)]
Show that the Wick quantization of the series in (3) is equal to the series of tree diagrams in (2).
\item[(5)]
Extend (2) and (3) to arbitrary times by using unitarity and conservation laws.
\end{itemize}

This paper is organized as follows. In Section \ref{section: Vlasov} we show that the classical Newtonian mechanics of point particles is the \emph{second quantization} of Vlasov theory, the latter being the mean-field (or ``classical'') limit of the former. The bulk of the paper is devoted to a rigorous analysis of the mean-field limit of Bose gases. In section \ref{section: setup} we recall some important concepts of quantum many-body theory and introduce a general formalism which is convenient when dealing with quantum gases. Section \ref{Schwinger-Dyson} contains an implementation of step (1) above. The convergence of the Schwinger-Dyson series for bounded interaction potentials is briefly discussed in Section \ref{section: bounded interaction}. Section \ref{section: Coulomb} implements step (2) above. Steps (3), (4) and (5) are implemented in Section \ref{section: mean-field limit for bosons}. Finally, Section \ref{section: generalizations} extends our results to more general interaction potentials as well as nonvanishing external potentials.

\noindent
\textbf{Acknowledgments}.\ We thank W.\ De Roeck, S.\ Graffi and A.\ Pizzo for useful discussions and encouragement. We would also like to thank a referee for pointing out Ref.\ \cite{LiebSeiringer} in connection with the remark following Corollary \ref{cor: density matrices}.

\section{Mean-field limit in classical mechanics} \label{section: Vlasov}
In this section we consider the example of classical Newtonian mechanics to illustrate how the atomistic constitution of matter arises by quantization of a continuum theory.  The aim of this section is to give a brief and nonrigorous overview of some ideas that we shall develop in the context of quantum Bose gases, in full detail, in the following sections.

A classical gas is described as a continuous medium whose state is given by a nonnegative mass density $\md \mu (x,v) = M f(x,v) \, \md x \, \md v$ on the ``one-particle'' phase space $\R^3 \times \R^3$. Here $M$ is the mass of one ``mole'' of gas; $\mu(A)$ is the mass of gas in the phase space volume $A \subset \R^3 \times \R^3$. Let $\int \md x\, \md v \, f(x,v) = \nu < \infty$ denote the number of ``moles'' of the gas, so that the total mass of the gas is $\mu(\R^3 \times \R^3) = \nu M$. An example of an equation of motion for $f(x,v)$ is the \emph{Vlasov equation}
\begin{equation} \label{Vlasov}
\partial_t f_t(x,v) \;=\; - \pb{v \cdot \nabla_x f_t}(x,v) + \frac{1}{m} \, \pb{\nabla V_{\text{eff}}[f_t] \cdot \nabla_v f_t}(x,v)\,,
\end{equation}
where $m$ is a constant with the dimension of a mass, $t$ denotes time, and
\begin{equation*}
V_{\text{eff}}[f](x) \;=\; V(x) + \int \md y \; W(x-y) \int \md v \; f(y,v)\,.
\end{equation*}
Here $V$ is the potential of external forces acting on the gas and $W$ is a (two-body) potential describing self-interactions of the gas.

The Vlasov equation arises as the mean-field limit of a classical Hamiltonian system of $n$ point particles of
mass $m$, with trajectories $(x_i(t))_{i = 1}^n$, moving in an external potential $V$ and interacting through two-body forces with potential $N^{-1} \, W(x_i - x_j)$. Here $N$ is the inverse coupling constant. We interpret $N$ as ``Avogadro's number'', i.e.\ as the number of particles per ``mole'' of gas. Thus, $M = m N$ and $n = \nu N$. More precisely, it is well-known (see \cite{BraunHepp, Neunzert}) that, under some technical assumptions on $V$ and $W$,
\begin{equation}
f_t(x,v) \;=\; \wstarlim_{n \to \infty} \frac{\nu}{n} \sum_{i = 1}^n \delta(x - x_i(t))  \, \delta(v - \dot{x}_i(t))
\end{equation}
exists for all times $t$ and is the (unique) solution of \eqref{Vlasov}, provided that this holds at time $t = 0$. Here, $f_t$ is viewed as an element of the dual space of continuous bounded functions. 

Note that $n$ and $N$ are, a priori, unrelated objects. While $n$ is the number of particles in the classical Hamiltonian system, $N^{-1}$ is by definition the coupling constant. The mean-field limit is the limit $n \to \infty$ while keeping $n \propto N$; the proportionality constant is $\nu$.

It is of interest to note that the Vlasov dynamics \eqref{Vlasov} may be interpreted as a Hamiltonian dynamics on an infinite-dimensional affine phase space $\Gamma_{\text{Vlasov}}$. To see this, we write
\begin{equation*}
f(x,v) \;=\; \bar{\alpha}(x,v) \alpha(x,v)\,,
\end{equation*}
where $\bar{\alpha}(x,v), \alpha(x,v)$ are complex coordinates on $\Gamma_{\text{Vlasov}}$.
For our purposes it is enough to say that $\Gamma_{\text{Vlasov}}$ is some dense subspace of $\LLL^2(\R^6)$ (typically a weighted Sobolev space of index 1). On $\Gamma_{\text{Vlasov}}$ we define a symplectic form through
\begin{equation*}
\omega \;=\; \mi \int \md x \, \md v \; \md \bar{\alpha}(x,v) \wedge \md \alpha(x,v)\,.
\end{equation*}
This yields a Poisson bracket which reads
\begin{align}
\hb{\alpha(x,v), \alpha(y, w)} &\;=\; \hb{\bar{\alpha}(x,v), \bar{\alpha}(y, w)} \;=\; 0\,,
\notag \\ \label{Vlasov Poisson bracket}
\hb{\alpha(x,v), \bar{\alpha}(y, w)} &\;=\; \mi \delta(x - y) \delta(v - w)\,.
\end{align}
A Hamilton function $H$ is defined on $\Gamma_{\text{Vlasov}}$ through
\begin{multline} \label{Vlasov Hamiltonian}
H(\alpha) \;\deq\; \mi \int \md x \, \md v \; \bar{\alpha}(x,v) \qbb{-v \cdot \nabla_x + \frac{1}{m} \,\nabla V(x) \cdot \nabla_v} \alpha(x,v)
\\
{}+{} \frac{\mi}{m} \int \md x \, \md v \; \bar{\alpha}(x,v) \qbb{\int \md y \, \md w \; \nabla W(x-y) \, \abs{\alpha(y,w)}^2} \cdot \nabla_v \alpha(x,v)\,.
\end{multline}
Note that $H$ is invariant under gauge transformations $\alpha \mapsto \me^{- \mi \theta} \alpha$, $\bar{\alpha} \mapsto \me^{\mi \theta} \bar{\alpha}$, which by Noether's theorem implies that $\int \abs{\alpha}^2 \, \md x \, \md v = \int f \, \md x \, \md v$ is conserved.

Let us abbreviate $K \deq - \nabla V / m$ and $F \deq - \nabla W / m$.
After a short calculation using \eqref{Vlasov Poisson bracket} we find that the Hamiltonian equation of motion $\dot{\alpha}_t(x,v) \;=\; \{H, \alpha_t(x,v)\}$ reads
\begin{multline} \label{Vlasov-Hamilton}
\dot{\alpha}_t(x,v) \;=\; \pb{-v \cdot \nabla_x - K(x) \cdot \nabla_v} \alpha_t(x,v) - \int \md y \, \md w \; F(x-y) \, \abs{\alpha_t(y,w)}^2 \cdot \nabla_v \alpha_t(x,v) 
\\
{}+{} \int \md y \, \md w \; F(x-y) \, \bar{\alpha}_t(y,w) \alpha_t(x,v) \cdot \nabla_w \alpha_t(y,w)\,. 
\end{multline}
Also, $\bar{\alpha}_t$ satisfies the complex conjugate equation. Therefore,
\begin{multline} \label{Intermediate Vlasov step}
\frac{\md}{\md t} \abs{\alpha_t(x,v)}^2 \;=\; \pb{-v \cdot \nabla_x - K(x) \cdot \nabla_v} \abs{\alpha_t(x,v)}^2 - \int \md y \, \md w \; F(x-y) \, \abs{\alpha_t(y,w)}^2 \cdot \nabla_v \abs{\alpha_t(x,v)}^2
\\
{}+{} \abs{\alpha_t(x,v)}^2 \int \md y \, \md w \; F(x-y) \cdot \qb{\bar{\alpha}_t(y,w) \nabla_w \alpha_t(y,w) + \alpha_t(y,w) \nabla_w \bar{\alpha}_t(y,w)}\,.
\end{multline}
We assume that
\begin{equation}\label{assumption for partial integration}
\abs{\alpha(x,v)} \;=\; o(\abs{(x,v)}^{-1})\,, \qquad (x,v) \to \infty\,.
\end{equation}
We shall shortly see that this property is preserved under time-evolution. By integration by parts, we see that the second line of \eqref{Intermediate Vlasov step} vanishes, and we recover the Vlasov equation of motion \eqref{Vlasov} for $f = \abs{\alpha}^2$.

We comment briefly on the existence and uniqueness of solutions to the Hamiltonian equation of motion \eqref{Vlasov-Hamilton}. Following Braun and Hepp \cite{BraunHepp}, we assume that $K$ and $F$ are bounded and continuously differentiable with bounded derivatives. We use polar coordinates
\begin{equation*}
\alpha \;=\; \beta\,  \me^{\mi \varphi}\,,
\end{equation*}
where $\varphi \in \R$ and $\beta \geq 0$\,. Then the Hamiltonian equation of motion \eqref{Vlasov-Hamilton} reads
\begin{subequations} \label{polar Vlasov-Hamilton}
\begin{align} \label{polar Vlasov-Hamilton 1}
\dot{\beta}_t(x,v) &\;=\; \pb{-v \cdot \nabla_x - K(x) \cdot \nabla_v} \beta_t(x,v) - \int \md y \, \md w \; F(x-y) \, \beta^2_t(y,w) \cdot \nabla_v \beta_t(x,v)
\\
\dot{\varphi}_t(x,v) &\;=\;  \pb{-v \cdot \nabla_x - K(x) \cdot \nabla_v} \varphi_t(x,v) - \int \md y \, \md w \; F(x-y) \, \beta^2_t(y,w) \cdot \nabla_v \varphi_t(x,v)
\notag \\ \label{polar Vlasov-Hamilton 2}
&\qquad {}+{}  \int \md y \, \md w \; F(x-y) \, \beta^2_t(y,w) \cdot \nabla_w \varphi_t(y,w)\,.
\end{align}
\end{subequations}
We consider two cases.
\begin{itemize}
\item[(i)] $\varphi = 0$. In this case $\alpha = \beta$ and the equations of motion \eqref{polar Vlasov-Hamilton} are equivalent to the Vlasov equation for $f = \beta^2$. The results of \cite{BraunHepp, Neunzert} then yield a global well-posedness result.
\item[(ii)] $\varphi \neq 0$. The equation of motion \eqref{polar Vlasov-Hamilton 1} is independent of $\varphi$. Case (i) implies that it has a unique global solution. In order to solve the linear equation \eqref{polar Vlasov-Hamilton 2}, we apply a contraction mapping argument. Consider the space $X \deq \{\varphi \in \CCC(\R^6) \,:\, \nabla \varphi \in \LLL^\infty(\R^6)\}$. Using Sobolev inequalities one finds that $X$, equipped with the norm $\norm{\varphi}_X \deq \abs{\varphi(0)} + \norm{\nabla \varphi}_\infty$, is a Banach space. We rewrite \eqref{polar Vlasov-Hamilton 2} as an integral equation, and using standard methods show that, for small times, it has a unique solution. Using conservation of $\int \md x \, \md v \; \beta^2_t$ we iterate this procedure to find a global solution. We omit further details.
\end{itemize}

Note that, as shown in \cite{BraunHepp}, the solution $\beta_t$ can be written using a flow $\phi^t$ on the one-particle phase space: $\beta_t(x,v) = \beta_0(\phi^{-t}(x,v))$. The flow $\phi^t(x,v) = (x(t), v(t))$ satisfies
\begin{align*}
\dot{x}(t) &\;=\; v(t)\,,
\\
\dot{v}(t) &\;=\; K(x(t)) + \int \md y \, \md w \; \beta_t^2(y,w) \, F(x(t) - y)\,.
\end{align*}
Using conservation of $\int \md x \, \md v \, \beta^2_t$ we find that there is a constant $C$ such that $\abs{\phi^{-t}(x,v)} \leq \abs{(x,v)} (1+t) + C (1+ t^2)$\,. Therefore \eqref{assumption for partial integration} holds for all times $t$ provided that it holds at time $t = 0$.

The Hamiltonian formulation of Vlasov dynamics can serve as a starting point to recover the atomistic Hamiltonian mechanics of point particles by quantization: Replace
\begin{equation*}
\bar{\alpha}(x,v) \;\rightarrow\; \widehat{\alpha}_N^*(x,v) \,, \qquad
\alpha(x,v) \;\rightarrow\; \widehat{\alpha}_N(x,v)\,,
\end{equation*}
where $\widehat{\alpha}_N^*$ and $ \widehat{\alpha}_N$ are creation and annihilation operators acting on the bosonic Fock space $\mathcal{F}_+\pb{\LLL^2(\R^6)}$; see Appendix \ref{second quantization}. They satisfy the canonical commutation relations \eqref{anticommutation relations}; explicitly,
\begin{align}
\qb{\widehat{\alpha}_N(x,v), \widehat{\alpha}_N(y, w)} &\;=\; \qb{\widehat{\alpha}_N^*(x,v), \widehat{\alpha}_N^*(y, w)} \;=\; 0\,,
\notag \\ \label{Vlasov CCR}
\qb{\widehat{\alpha}_N(x,v), \widehat{\alpha}_N^*(y, w)} &\;=\; \frac{1}{N} \, \delta(x - y) \delta(v - w)\,.
\end{align}
Given a function $A$ on $\Gamma_{\text{Vlasov}}$ which is a polynomial in $\ol{\alpha}$ and $\alpha$, we define an operator $\widehat{A}_N$ on $\mathcal{F}_+$ by replacing $\alpha^\#$ with $\widehat{\alpha}_N^\#$ and Wick-ordering the resulting expression. We denote this quantization map by $\widehat{(\cdot)}_N$. Here, $N^{-1}$ is the deformation parameter of the quantization: We find that
\begin{equation*}
\qb{\widehat{A}_N, \widehat{B}_N} \;=\; \frac{N^{-1}}{\mi} \widehat{\h{A,B}}_N + O(N^{-2})\,,
\end{equation*}
for $N \to \infty$. Here $A$ and $B$ are polynomial functions on $\Gamma_{\text{Vlasov}}$.

The dynamics of a state $\Phi \in \mathcal{F}$ is given by the Schr\"odinger equation
\begin{equation} \label{schrodinger equation for classical dynamics}
\mi N^{-1} \partial_t \Phi_t \;=\; \widehat{H}_N \Phi_t\,,
\end{equation}
where $\widehat{H}_N$ is the quantization of the Vlasov Hamiltonian $H$.
In order to identify the dynamics given by \eqref{schrodinger equation for classical dynamics} with the classical dynamics of point particles, we study wave functions $\Phi^{(n)}(x_1, v_1, \dots, x_n, v_n)$ in the $n$-particle sector of $\mathcal{F}_+$, and interpret $\rho^{(n)} \deq \abs{\Phi}^2$ as a probability density on the $n$-body classical phase space.
If $\Omega \in \mathcal{F}_+$ denotes the vacuum vector annihilated by $\widehat{\alpha}_N(x,v)$ then
\begin{equation*}
\Phi^{(n)} \;=\; \frac{N^{n/2}}{\sqrt{n!}} \int \md x_1 \, \md v_1 \cdots \md x_n \, \md v_n \; \Phi^{(n)}(x_1, v_1, \dots, x_n,v_n) \, \widehat{\alpha}_N^*(x_n,v_n) \cdots \widehat{\alpha}_N^*(x_1, v_1) \, \Omega\,.
\end{equation*}
It is a simple matter to check that \eqref{Vlasov CCR} and \eqref{schrodinger equation for classical dynamics} imply that
\begin{equation*}
\partial_t \Phi^{(n)}_t \;=\; \sum_{i = 1}^n \qbb{-v_i \cdot \nabla_{x_i} + \frac{1}{m} \, \nabla V(x_i) \cdot \nabla_{v_i}} \Phi^{(n)}_t + \frac{1}{N} \sum_{1 \leq i \neq j \leq n} \frac{1}{m} \, \nabla W(x_i - x_j) \cdot \nabla_{v_i} \Phi^{(n)}_t\,. 
\end{equation*}
Also, $\ol{\Phi^{(n)}_t}$ satisfies the same equation.
Therefore,
\begin{equation*}
\partial_t \rho^{(n)}_t \;=\; \sum_{i = 1}^n \qbb{- v_i \cdot \nabla_{x_i} + \frac{1}{m} \, \nabla V(x_i) \cdot \nabla_{v_i}} \rho^{(n)}_t + \frac{1}{N} \sum_{1 \leq i \neq j \leq n} \frac{1}{m} \, \nabla W(x_i - x_j) \cdot \nabla_{v_i} \rho^{(n)}_t\,. 
\end{equation*}
This is the Liouville equation corresponding to the Hamiltonian equations of motion of $n$ classical point particles,
\begin{align*}
\partial_t x_i &\;=\; v_i\,,
\\
m \, \partial_t v_i &\;=\; -\nabla V(x_i) - \frac{1}{N} \sum_{j \neq i} \nabla W (x_i - x_j)\,.
\end{align*}

Analogous results can be proven if $\widehat{\alpha}_N^*$ and $\widehat{\alpha}_N$ are chosen to be fermionic creation and annihilation operators obeying the canonical anti-commutation relations and acting on the fermionic Fock space $\mathcal{F}_-(\LLL^2(\R^6))$.

\section{Quantum gases: the setup} \label{section: setup}
Although our main results are restricted to bosons, all of the following rather general formalism remains unchanged for fermions. We therefore consider both bosonic and fermionic statistics throughout Sections \ref{section: setup} -- \ref{section: Coulomb}. Details on systems of fermions will appear elsewhere.

Throughout the following we consider the one-particle Hilbert space 
\begin{equation*}
\mathcal{H} \deq \LLL^2(\R^3, \md x)\,.
\end{equation*}
We refer the reader to Appendix \ref{second quantization} for our choice of notation and a short discussion of many-body quantum mechanics.

In the following a central role is played by the $p$-particles operators, i.e.\ closed operators $a^{(p)}$ on $\mathcal{H}_\pm^{(p)} = P_{\pm} \mathcal{H}^{\otimes p}$, where $P_+$ and $P_-$ denote symmetrization and anti-symmetrization, respectively. When using second-quantized notation it is convenient to use the operator kernel of $a^{(p)}$. Here is what this means (see  \cite{ReedSimonI} for details): Let $\mathscr{S}(\R^d)$ be the usual Schwartz space of smooth functions of rapid decrease, and $\mathscr{S}'(\R^d)$ its topological dual. The nuclear theorem states that to every operator $A$ on $\LLL^2(\R^d)$, such that the map $\pb{f,g} \mapsto \scalarb{f}{A g}$ is separately continuous on $\mathscr{S}(\R^d) \times \mathscr{S}(\R^d)$, there belongs a tempered distribution (``kernel'') $\tilde{A} \in \mathscr{S}'(\R^{2d})$, such that 
\begin{equation*}
\scalar{f}{A g} \;=\; \tilde{A} (\bar{f} \otimes g)\,.
\end{equation*}
In the following we identify $\tilde{A}$ with $A$. In the suggestive physicist's notation we thus have
\begin{equation*}
\scalarb{f}{a^{(p)} g} \;=\; \int \md x_1 \cdots \md x_p \, \md y_1 \cdots \md y_p \; 
\ol{f}(x_1, \dots, x_p) \, a^{(p)}(x_1, \dots, x_p;y_1, \dots, y_p) \, g(y_1, \dots, y_p)\,,
\end{equation*}
where $f,g \in \mathscr{S}(\R^{3p})$.
It will be easy to verify that all $p$-particle operators that appear in the following satisfy the above condition; this is for instance the case for all bounded $a^{(p)} \in \mathcal{B}(\mathcal{H}^{\otimes p})$.

%\footnote{Such an object is rigorously defined as a sesquilinear form on the space $\{\Phi \in \mathcal{F}^0_\pm \,:\, \Phi^{(n)} \in \mathscr{S}(\R^{3n}) \; \forall n\}$, on which it is closable.}
Next, we define second quantization $\widehat{\A}_N$. It maps a closed operator on $\mathcal{H}^{(p)}_\pm$ to a closed operator on $\mathcal{F}_\pm$ according to the formula
\begin{multline}
\widehat{\A}_N(a^{(p)}) \;\deq\; \int \md x_1 \cdots \md x_p \, \md y_1 \cdots \md y_p 
\\
\widehat{\psi}_N^*(x_p) \cdots \widehat{\psi}_N^*(x_1) \, a^{(p)}(x_1, \dots, x_p;y_1, \dots, y_p) \, 
\widehat{\psi}_N(y_1) \cdots \widehat{\psi}_N(y_p)\,.
\end{multline}
Here $\widehat{\psi}^\#_N \;\deq\; \frac{1}{\sqrt{N}} \, \widehat{\psi}^\#$, where $\widehat{\psi}^\#$ is the usual creation or annihilation operator; see Appendix \ref{second quantization}.

In order to understand the action of $\widehat{\A}_N(a^{(p)})$ on $\mathcal{H}^{(n)}_\pm$, we write
\begin{equation*}
\Phi^{(n)} \;=\; \frac{N^{n/2}}{\sqrt{n!}}\int \md z_1 \cdots \md z_n \; \Phi^{(n)}(z_1, \dots, z_n) \, \widehat{\psi}^*_N(z_n) \cdots \widehat{\psi}^*_N(z_1) \, \Omega
\end{equation*}
and apply $\widehat{\A}_N(a^{(p)})$ to the right side. By using the (anti)commutation relations \eqref{anticommutation relations} to pull the $p$ annihilation operators $\widehat{\psi}_N(y_i)$ through the $n$ creation operators $\widehat{\psi}_N^*(z_i)$, and $\widehat{\psi}_N(x) \, \Omega = 0$, we get the ``first quantized'' expression
\begin{equation}
\widehat{\A}_N(a^{(p)})\bigr|_{\mathcal{H}^{(n)}_\pm} \;=\;
\begin{cases}
\frac{p!}{N^p} \binom{n}{p} P_{\pm} (a^{(p)} \otimes \umat^{(n-p)}) P_{\pm}\,, &n\geq p\,,
\\
0\,, &n<p\,.
\end{cases}
\end{equation}
This may be viewed as an alternative definition of $\widehat{\A}_N(a^{(p)})$.

We define $\widehat{\mathfrak{A}}$ as the linear span of $\hb{\widehat{\A}_N(a^{(p)}) \,:\, p \in \N,\, a^{(p)} \in \mathcal{B}(\mathcal{H}^{(p)}_\pm)}$. Then $\widehat{\mathfrak{A}}$ is a $*$-algebra of closable operators on $\mathcal{F}^0_\pm$ (see Appendix \ref{second quantization}). We list some of its important properties, whose straightforward proofs we omit.
\begin{itemize}
\item[(i)]
$\widehat{\A}_N(a^{(p)})^* = \widehat{\A}_N((a^{(p)})^*)$\,.
\item[(ii)]
If $a^{(p)} \in \mathcal{B}(\mathcal{H}^{(p)}_\pm)$ and $b^{(q)} \in \mathcal{B}(\mathcal{H}^{(q)}_\pm)$, then 
\begin{equation} \label{product of two second quantized operators}
\widehat{\A}_N(a^{(p)}) \, \widehat{\A}_N(b^{(q)}) \;=\; \sum_{r = 0}^{\min(p,q)} \binom{p}{r} \binom{q}{r} \frac{r!}{N^r}\, \widehat{\A}_N\pb{a^{(p)} \bullet_r b^{(q)}}\,,
\end{equation}
where
\begin{equation}
a^{(p)} \bullet_r b^{(q)} \;\deq\; P_\pm \,(a^{(p)} \otimes \umat^{(q-r)}) \, (\umat^{(p-r)} \otimes b^{(q)}) \,P_\pm \;\in\; \mathcal{B}(\mathcal{H}^{(p+q-r)}_\pm)\,.
\end{equation}
\item[(iii)]
The operator $\widehat{\A}(a^{(p)})$ leaves the $n$-particle subspaces $\mathcal{H}^{(n)}_\pm$ invariant.
\item[(iv)]
If $a^{(p)} \in \mathcal{B}(\mathcal{H}^{(p)}_\pm)$ and $b\in \mathcal{B}(\mathcal{H})$ is invertible, then 
\begin{equation}
\Gamma(b^{-1}) \, \widehat{\A}_N(a^{(p)}) \, \Gamma(b) \;=\; \widehat{\A}_N \pb{(b^{-1})^{\otimes p} \, a^{(p)} \, b^{\otimes p}}\,,
\end{equation}
where $\Gamma(b)$ is defined on $\mathcal{H}^{(n)}_\pm$ by $b^{\otimes n}$.
\item[(v)]
If $a^{(p)} \in \mathcal{B}(\mathcal{H}^{(p)}_\pm)$ then 
\begin{equation} \label{norm of second quantized observable}
\normB{\widehat{\A}_N(a^{(p)}) \bigr|_{\mathcal{H}^{(n)}_\pm}} \;\leq\; \pbb{\frac{n}{N}}^p \norm{a^{(p)}}\,.
\end{equation}
\end{itemize}
Of course, on an appropriate dense domain, \eqref{product of two second quantized operators} holds for unbounded operators $a^{(p)}$ and $b^{(q)}$ too. We introduce the notation
\begin{equation}
\qb{a^{(p)}, b^{(q)}}_r \;\deq\; a^{(p)} \bullet_r b^{(q)} - b^{(q)} \bullet_r a^{(p)}\,.
\end{equation}
Note that $\qb{a^{(p)}, b^{(q)}}_0 = 0$. Thus,
\begin{equation} \label{commutator of second quantized operators}
\qb{\widehat{\A}_N(a^{(p)}), \widehat{\A}_N(b^{(q)})} \;=\; \sum_{r = 1}^{\min(p,q)} \binom{p}{r} \binom{q}{r} \frac{r!}{N^r}\, \widehat{\A}_N\pb{\qb{a^{(p)} , b^{(q)}}_r}\,.
\end{equation}

We now move on to discuss dynamics. Take a one-particle Hamiltonian $h^{(1)} \equiv h$ of the form $h = -\Delta + v$, where $\Delta$ is the Laplacian over $\R^3$ and $v$ is some real function. We denote by $V$ the multiplication operator $v(x)$. Two-body interactions are described by a real, even function $w$ on $\R^3$. This induces a two-particle operator $W^{(2)} \equiv W$ on $\mathcal{H}^{\otimes 2}$, defined as the multiplication operator $w(x_1 - x_2)$.
We define the Hamiltonian
\begin{equation} \label{quantum Hamiltonian}
\widehat{H}_N \;\deq\; \widehat{\A}_N(h) + \frac{1}{2} \, \widehat{\A}_N(W)\,.
\end{equation}
Under suitable assumptions on $v$ and $w$ that we make precise in the following sections, one shows that $\widehat{H}_N$ is a well-defined self-adjoint operator on $\mathcal{F}_\pm$. It is convenient to introduce $H_N \deq N \widehat{H}_N$. On $\mathcal{H}^{(n)}_\pm$ we have the ``first quantized'' expression
\begin{equation}
H_N \big|_{\mathcal{H}^{(n)}_\pm} \;=\; \sum_{i = 1}^n h_i + \frac{1}{N}\sum_{1 \leq 1 < j \leq n} W_{ij} \;\eqd\; H_0 + \frac{1}{N} W\,,
\end{equation}
in self-explanatory notation.

\section{Schwinger-Dyson expansion and loop counting} \label{Schwinger-Dyson}
Without loss of generality, we assume throughout the following that $t \geq 0$.

Let $a^{(p)} \in \mathcal{B}(\mathcal{H}^{(p)}_\pm)$ and $w$ be bounded, i.e.\ $w \in \LLL^\infty(\R^3)$. Using the fundamental theorem of calculus and the fact that the unitary group $(\me^{-\mi t H_0})_{t \in \R}$ is strongly differentiable one finds 
\begin{align*}
&\me^{\mi t H_N} \, \widehat{\A}_N(a^{(p)}) \, \me^{- \mi t H_N} \, \Phi^{(n)}
\\
&\quad\;=\; \me^{\mi s H_N} \me^{- \mi s H_0} \me^{\mi t H_0} \, \widehat{\A}_N(a^{(p)}) \, \me^{- \mi t H_0} \me^{\mi s H_0} \me^{-\mi s H_N} \, \Phi^{(n)} \bigr|_{s = t}
\\
&\quad\;=\; \widehat{\A}_N(a^{(p)}_t) \, \Phi^{(n)} + \int_0^t \md s \; \me^{\mi s H_N} \me^{-\mi sH_0} \frac{\mi N }{2} \qB{\widehat{\A}_N(W_s), \widehat{\A}_N(a^{(p)}_t)} \, \me^{\mi sH_0} \me^{-\mi s H_N}\, \Phi^{(n)}\,,
\end{align*}
where $(\cdot)_t \deq \Gamma(\me^{\mi th}) (\cdot) \Gamma(\me^{-\mi th})$ denotes free time evolution. As an equation between operators defined on $\mathcal{F}^0_\pm$, this reads
\begin{equation} \label{Duhamel}
\me^{\mi t H_N} \, \widehat{\A}_N (a^{(p)}) \, \me^{- \mi t H_N} \;=\; \widehat{\A}_N(a^{(p)}_t) + \int_0^t \md s \; \me^{\mi s H_N} \me^{-\mi sH_0} \frac{\mi N}{2} \qB{\widehat{\A}_N(W_s), \widehat{\A}_N(a^{(p)}_t)} \, \me^{\mi sH_0} \me^{-\mi s H_N}\,.
\end{equation}
Iteration of \eqref{Duhamel} yields the formal power series
\begin{equation} \label{Schwinger-Dyson series}
\sum_{k = 0}^\infty \int_{\Delta^k(t)} \md \ul{t} \; \frac{(\mi N)^k}{2^k} \qB{\widehat{\A}_N(W_{t_k}), \dots \qB{\widehat{\A}_N(W_{t_1}), \widehat{\A}_N(a^{(p)}_t)}\dots}\,.
\end{equation}
It is easy to see that, on $\mathcal{H}^{(n)}_\pm$, the $k$-th term of \eqref{Schwinger-Dyson series} is bounded in norm by
\begin{equation} \label{simple bound on k'th term}
\frac{\pb{t n^2 \norm{w}_\infty/N}^k}{k!} \, \pbb{\frac{n}{N}}^p \norm{a^{p}}\,.
\end{equation}
Therefore, on $\mathcal{H}^{(n)}_\pm$, the series \eqref{Schwinger-Dyson series} converges in norm for all times. Furthermore, \eqref{simple bound on k'th term} implies that the rest term arising from the iteration of \eqref{Duhamel} vanishes for $k \to \infty$, so that \eqref{Schwinger-Dyson series} is equal to \eqref{Duhamel}.

The \emph{mean-field limit} is the limit $n = \nu N \to \infty$, where $\nu >0$ is some constant.
The above estimate is clearly inadequate to prove statements about the mean-field limit. In order to obtain estimates uniform in $N$, more care is needed.

To see why the above estimate is so crude, consider the commutator
\begin{equation*}
\frac{\mi N}{2} \qB{\widehat{\A}_N(W_s), \widehat{\A}_N(a^{(p)}_t)} \Bigr|_{\mathcal{H}^{(n)}_\pm} \;=\; 
\frac{p!}{N^p} \binom {n}{p} \frac{\mi}{N} P_\pm \sum_{1 \leq i < j \leq n} \qb{W_{ij,s}, a^{(p)}_t \otimes \umat^{(n-p)}} P_\pm\,.
\end{equation*}
We see that most terms of the commutator vanish (namely, whenever $p<i<j$). Thus, for large $n$, the above estimates are highly wasteful. This can be remedied by more careful bookkeeping. We split the commutator into two terms: the \emph{tree terms}, defined by $1 \leq i \leq p$ and $p+1 \leq j \leq n$, and the \emph{loop terms}, defined by $1 \leq i < j \leq p$. All other terms vanish. This splitting can also be inferred from \eqref{commutator of second quantized operators}.

The naming originates from a diagrammatic representation (see Figure \ref{figure: tree and loop}). A $p$-particle operator is represented as a wiggly vertical line to which are attached $p$ horizontal branches on the left and $p$ horizontal branches on the right. Each branch on the left represents a creation operator $\widehat{\psi}_N^*(x_i)$, and each branch on the right an annihilation operator $\widehat{\psi}_N(y_i)$. The product $\widehat{\A}_N(a^{(p)}) \widehat{\A}_N(b^{(q)})$ of two operators is given by the sum over all possible pairings of the annihilation operators in $\widehat{\A}_N(a^{(p)})$ with the creation operators in $\widehat{\A}_N(b^{(q)})$. Such a contraction is graphically represented as a horizontal line joining the corresponding branches. We consider diagrams that arise in this manner from the multiplication of a finite number of operators of the form $\widehat{\A}_N(a^{(p)})$. 
\begin{figure}[ht!]
\vspace{0.5cm}
\begin{center}
\psfrag{1}[][]{$a^{(p)}_t$}
\psfrag{2}[][]{$W_{i\,p+1,s}$}
\psfrag{3}[][]{$a^{(p)}_t$}
\psfrag{4}[][]{$W_{ij,s}$}
\includegraphics[width = 10cm]{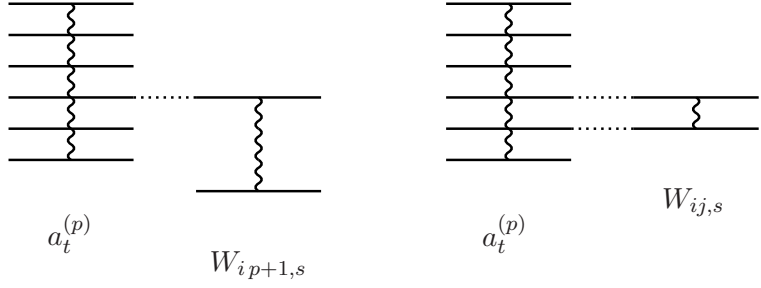}
\end{center}
\caption{\textit{Two terms of the product $\widehat{\A}_N(a^{(p)}_t) \, \widehat{\A}_N(W_s)$, represented as labelled diagrams. A tree term (left) produces a tree diagram. A loop term (right) produces a diagram with one loop.} \label{figure: tree and loop}}
\end{figure}

We now generalize this idea to a systematic scheme for the multiple commutators appearing in the Schwinger-Dyson expansion. To this end, we decompose the multiple commutator
\begin{equation*}
\frac{(\mi N)^k}{2^k} \qB{\widehat{\A}_N(W_{t_k}), \dots \qB{\widehat{\A}_N(W_{t_1}), \widehat{\A}_N(a^{(p)}_t)}\dots}
\end{equation*}
into a sum of $2^k$ terms obtained by writing out each commutator. Each resulting term is a product of $k+1$ second-quantized operators, which we furthermore decompose into a sum over all possible contractions for which $r > 0$ in \eqref{product of two second quantized operators} (at least one contraction for each multiplication). The restriction $r > 0$ follows from $[a^{(p)}, b^{(q)}]_0 = 0$. This is equivalent to saying that all diagrams are connected.

We call the resulting terms \emph{elementary}. The idea is to classify all elementary terms according to their number of loops $l$. Write
\begin{equation} \label{splitting of multiple commutator}
\frac{(\mi N)^k}{2^k} \qB{\widehat{\A}_N(W_{t_k}), \dots \qB{\widehat{\A}_N(W_{t_1}), \widehat{\A}_N(a^{(p)}_t)}\dots} \;=\; \sum_{l = 0}^k \frac{1}{N^l} \, \widehat{\A}_N \pb{ G^{(k,l)}_{t,t_1,\dots,t_k}(a^{(p)})}\,,
\end{equation}
where $G^{(k,l)}_{t,t_1,\dots,t_k}(a^{(p)})$ is a $(p+k-l)$-particle operator, equal to the sum of all elementary terms with $l$ loops. It is defined through the recursion relation (on $\mathcal{H}^{(p+k-l)}_\pm$)
\begin{align}
G^{(k,l)}_{t,t_1,\dots,t_k}(a^{(p)}) 
&\;=\; \mi (p+k-l-1) \qB{W_{t_k} , G^{(k-1,l)}_{t, t_1, \dots, t_{k - 1}}(a^{(p)})}_1
\notag \\
&\qquad {}+{} \mi \binom{p+k-l}{2} \qB{W_{t_k}, G^{(k-1, l-1)}_{t, t_1, \dots, t_{k - 1}}(a^{(p)})}_2
\notag \\
&\;=\; \mi P_\pm \sum_{i = 1}^{p+k-l-1} \qB{W_{i \, p+k-l, t_k}, G^{(k-1,l)}_{t, t_1, \dots, t_{k - 1}}(a^{(p)}) \otimes \umat} P_\pm 
\notag \\ \label{recursive definition of graphs}
&\qquad {}+{} \mi P_\pm  \sum_{1 \leq i < j \leq p+k-l} \qB{W_{ij, t_k}, G^{(k-1, l-1)}_{t, t_1, \dots, t_{k - 1}}(a^{(p)})} P_\pm \,,
\end{align}
as well as $G^{(0,0)}_t(a^{(p)}) \deq a^{(p)}_t$. If $l < 0$, $l > k$, or $p+k-l > n$ then $G^{(k,l)}_{t,t_1,\dots,t_k}(a^{(p)}) = 0$.
The interpretation of the recursion relation is simple: a $(k,l)$-term arises from either a $(k - 1, l)$-term without adding a loop or from a $(k - 1, l-1)$-term to which a loop is added. It is not hard to see, using induction on $k$ and the definition \eqref{recursive definition of graphs}, that \eqref{splitting of multiple commutator} holds.
It is often convenient to have an explicit formula for the decomposition into elementary terms:
\begin{equation*}
G^{(k,l)}_{t,t_1,\dots,t_k}(a^{(p)}) \;=\; \sum_{\alpha = 1}^{c(p,k,l)} G^{(k,l)(\alpha)}_{t,t_1,\dots,t_k}(a^{(p)})\,,
\end{equation*}
where $G^{(k,l)(\alpha)}_{t,t_1,\dots,t_k}(a^{(p)})$ is an elementary term, and $c(p,k,l)$ is the number of elementary terms in $G^{(k,l)}_{t,t_1,\dots,t_k}(a^{(p)})$. 
\begin{figure}[ht!]
\vspace{0.5cm}
\begin{center}
\psfrag{1}[][]{$0$}
\psfrag{2}[][]{$1$}
\psfrag{3}[][]{$2$}
\psfrag{4}[][]{$3$}
\psfrag{5}[][]{$4$}
\includegraphics[width = 12cm]{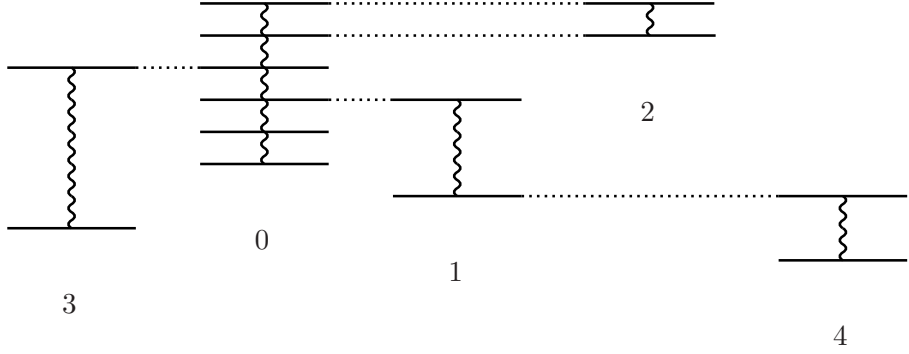}
\end{center}
\caption{\textit{The labelled diagram corresponding to a one-loop elementary term in the commutator of order $4$.} \label{figure: general graph}}
\end{figure}

In order to establish a one-to-one correspondence between elementary terms and diagrams, we introduce a labelling scheme for diagrams. Consider an elementary term arising from a choice of contractions in the multiple commutator of order $k$, along with its diagram.
We label all vertical lines $v$ with an index $i_v \in \N$ as follows. The vertical line of $a^{(p)}$ is labelled by $0$. The vertical line of the first (i.e.\ innermost in the multiple commutator) interaction operator is labelled by $1$, of the second by $2$, and so on (see Figure \ref{figure: general graph}). Conversely, every elementary term is uniquely determined by its labelled diagram. We consequently use $\alpha = 1, \dots, c(p,k,l)$ to index either elementary terms or labelled diagrams.

Use the shorthand $\ul{t} = (t_1, \dots, t_k)$ and define
\begin{equation}
G^{(k,l)}_{t}(a^{(p)}) \;\deq\; \int_{\Delta^k(t)} \md \ul{t} \; G^{(k,l)}_{t,\ul{t}}(a^{(p)})\,.
\end{equation}
In summary, we have an expansion in terms of the number of loops $l$:
\begin{equation} \label{loop expansion}
\me^{\mi t H_N} \, \widehat{\A}_N (a^{(p)}) \, \me^{- \mi t H_N} \;=\; \sum_{k = 0}^\infty \sum_{l = 0}^k \frac{1}{N^l} \, \widehat{\A}_N \pb{G^{(k,l)}_t(a^{(p)})}\,,
\end{equation}
which converges in norm on $\mathcal{H}^{(n)}_\pm$, $n \in \N$, for all times $t$.

\section{Convergence for bounded interaction} \label{section: bounded interaction}
For a bounded interaction potential, $\norm{w}_\infty < \infty$, it is now straightforward to control the mean-field limit.
\begin{lemma} \label{lemma: bound for graph for bounded interaction}
We have the bound
\begin{equation} \label{definition of c(p,k,l)}
\normB{G^{(k,l)}_{t,\ul{t}}(a^{(p)})} \;\leq\; c(p,k,l) \norm{w}_\infty^k \, \norm{a^{(p)}}\,.
\end{equation}
Furthermore,
\begin{equation} \label{bound for c(p,k,l)}
c(p,k,l) \;\leq\; 2^k \binom{k}{l} \, (p + k - l)^l \, (p + k - 1) \cdots p\,.
\end{equation}
\end{lemma}
\begin{proof}
Assume first that $l = 0$. Then the number of labelled diagrams is clearly given by $2^k p \cdots (p+k-1)$. Now if there are $l$ loops, we may choose to add them at any $l$ of the $k$ steps when computing the multiple commutator. Furthermore, each addition of a loop produces at most $p + k - l$ times more elementary terms than the addition of a tree branch. Combining these observations, we arrive at the claimed bound for $c(p,k,l)$.

Alternatively, it is a simple exercise to show the claim, with $c(p,k,l)$ replaced by the bound \eqref{bound for c(p,k,l)}, by induction on $k$.
\end{proof}

\begin{lemma} \label{convergence of the loop expansion}
Let $\nu > 0$ and $t < (8 \nu \norm{w}_\infty)^{-1}$. Then, on $\mathcal{H}^{(\nu N)}_\pm$, the Schwinger-Dyson series \eqref{loop expansion} converges in norm, uniformly in $N$.

\end{lemma}
\begin{proof}
Recall that $p+k-l \leq n$ for nonvanishing $\widehat{\A}_N \pb{G^{(k,l)}_{t,\ul{t}}(a^{(p)})} \bigr|_{\mathcal{H}^{(n)}_\pm}$. Using the symbol $I_{\{A\}}$, defined as $1$ if $A$ is true and $0$ if $A$ is false, we find
\begin{align*}
&\sum_{k = 0}^\infty \sum_{l = 0}^k \frac{1}{N^l} \int_{\Delta^k(t)} \md \ul{t} \; \normB{\widehat{\A}_N \pb{G^{(k,l)}_{t,\ul{t}}(a^{(p)})} \bigr|_{\mathcal{H}^{(\nu N)}_\pm}} 
\\
&\qquad \leq\; \sum_{k = 0}^\infty \sum_{l = 0}^k \frac{(p+k-l)^l}{N^l} \,I_{\{p+k-l \leq \nu N\}}\, \frac{1}{k!} (2 \norm{w}_\infty t)^k \, \binom{k}{l} \binom{p+k-1}{k} k! \, \nu^{p+k-l} \, \norm{a^{(p)}}
\\
&\qquad\leq\; \sum_{k = 0}^\infty (8 \nu \norm{w}_\infty t)^k \,
(2\nu)^p \, \norm{a^{(p)}}
\\ 
&\qquad=\;  \frac{1}{1 - 8 \nu \norm{w}_\infty t} \, (2\nu)^p \, \norm{a^{(p)}}\,,
\end{align*}
where we used that $\sum_{l = 0}^k \binom{k}{l} = 2^k$, and in particular $\binom{k}{l} \leq 2^k$.
\end{proof}

In the spirit of semi-classical expansions, we can rewrite the Schwinger-Dyson series to get a ``$1/N$-expansion'', whereby all $l$-loop terms add up to an operator of order $O(N^{-l})$.
\begin{lemma} \label{convergence for bounded potential}
Let $t < (8 \nu \norm{w}_\infty)^{-1}$ and $L \in \N$. Then we have on $\mathcal{H}^{(\nu N)}_\pm$
\begin{equation*}
\me^{\mi t H_N} \, \widehat{\A}_N(a^{(p)}) \, \me^{- \mi t H_N} \;=\; \sum_{l = 0}^{L-1} \frac{1}{N^l} \sum_{k = l}^\infty \, \widehat{\A}_N \pb{G^{(k,l)}_t(a^{(p)})} + O \pbb{\frac{1}{N^L}}\,,
\end{equation*}
where the sum converges uniformly in $N$.
\end{lemma}
\begin{proof}
Instead of the full Schwinger-Dyson expansion \eqref{Schwinger-Dyson series}, we can stop the expansion whenever $L$ loops have been generated. More precisely, we iterate \eqref{Duhamel} and use \eqref{commutator of second quantized operators} at each iteration to split the commutator into tree ($r = 1$) and loop ($r = 2$) terms. Whenever a term obtained in this fashion has accumulated $L$ loops, we stop expanding and put it into a remainder term. Thus all fully expanded terms are precisely those arising from diagrams containing up to $L-1$ loops, and it is not hard to show that the remainder term is of order $N^{-L}$.

In view of later applications, we also give a proof using the fully expanded Schwinger-Dyson series.
From Lemma \ref{convergence of the loop expansion} we know that the sum converges on $\mathcal{H}_\pm^{(\nu N)}$ in norm, uniformly in $N$, and can be reordered as
\begin{equation*}
\me^{\mi t H_N} \, \widehat{\A}_N (a^{(p)}) \, \me^{- \mi t H_N} \;=\; \sum_{l = 0}^\infty \frac{1}{N^l} \sum_{k = l}^\infty \int_{\Delta^k(t)} \md \ul{t} \; \widehat{\A}_N \pb{G^{(k,l)}_{t,\ul{t}}(a^{(p)})}\,,
\end{equation*}
as an identity on $\mathcal{H}^{(\nu N)}_\pm$.
Proceeding as above we find
\begin{align*}
&\sum_{l = L}^\infty \frac{1}{N^l} \sum_{k = l}^\infty \int_{\Delta^k(t)} \md \ul{t} \; \normB{\widehat{\A}_N \pb{G^{(k,l)}_{t,\ul{t}}(a^{(p)})} \bigr|_{\mathcal{H}^{(\nu N)}_\pm}} 
\\
&\qquad\leq\;
\frac{1}{N^L} \sum_{l = L}^\infty  \sum_{k = l}^\infty \frac{(p+k-l)^l}{N^{l-L}} \,I_{\{p+k-l \leq \nu N\}} \, \frac{1}{k!} (2 \norm{w}_\infty t)^k \, \binom{k}{l} \binom{p+k-1}{k} k! \, \nu^{p+k-l} \, \norm{a^{(p)}}
\\
&\qquad\leq\;
\frac{1}{(\nu N)^L} \sum_{l = L}^\infty  \sum_{k = l}^\infty (p+k-l)^L (8 \nu \norm{w}_\infty t)^k \, (2\nu)^p\, \norm{a^{(p)}}
\\
&\qquad=\;
\frac{1}{(\nu N)^L} \sum_{l = L}^\infty  \sum_{k = 0}^\infty (p+k)^L (8 \nu \norm{w}_\infty t)^{k+l} \, (2\nu)^p\, \norm{a^{(p)}}
\\
&\qquad\leq\;
\frac{1}{(\nu N)^L} \sum_{l = L}^\infty  (8 \nu \norm{w}_\infty t)^l \frac{\me^p \, L!}{(1 - 8 \nu \norm{w}_\infty t)^{L+1}} (2 \nu)^p\, \norm{a^{(p)}}
\\
&\qquad=\;
\frac{1}{(\nu N)^L} \frac{\me^p \,L! \,(8 \nu \norm{w}_\infty t)^L}{(1 - 8 \nu \norm{w}_\infty t)^{L+2}} (2 \nu)^p\, \norm{a^{(p)}}\,,
\end{align*}
where we used that $\sum_{k = 0}^\infty (p+k)^L \, x^k \leq \frac{\me^p \, L!}{(1 - x)^{L+1}}$.
\end{proof}

\section{Convergence for Coulomb interaction} \label{section: Coulomb}
In this section we consider an interaction potential of the form
\begin{equation}
w(x) \;=\;\kappa \frac{1}{\abs{x}}\,,
\end{equation}
where $\kappa \in \R$.
We take the one-body Hamiltonian to be
\begin{equation*}
h \;=\; - \Delta\,,
\end{equation*}
the nonrelativistic kinetic energy without external potentials. We assume this form of $h$ and $w$
throughout Sections \ref{section: Coulomb} and \ref{section: mean-field limit for bosons}. In Section \ref{section: generalizations}, we discuss some generalizations.

\subsection{Kato smoothing}
The non-relativistic dispersive nature of the free time evolution $\ee^{\ii t\Delta}$ is essential for controlling singular potentials.
It embodied by the following dispersive estimate, which is sometimes referred to as Kato's smoothing estimate, as it was first derived using Kato's theory of smooth perturbations; see \cite{ReedSimonIV, Simon1992}. Here we present a new, elementary proof, which yields the sharp constant and may be easily generalized to free Hamiltonians of the form $(-\Delta)^\gamma$, where $1/2 < \gamma < d/2$ and $d$ denotes the number of spatial dimensions.

\begin{lemma} \label{lemma: Kato smoothing}
For $d \geq 3$ and $\psi \in \LLL^2(\R^d)$ we have
\begin{equation} \label{kato smoothing estimate}
\int \dd t \; \normb{\abs{x}^{-1} \, \ee^{\ii t \Delta} \, \psi}^2 \;\leq\; \frac{\pi}{d - 2} \norm{\psi}^2\,.
\end{equation}
More generally, for $d \geq 2$ and $\gamma$ satisfying $1/2 < \gamma < d/2$ we have
\begin{equation} \label{generalized Kato smoothing}
\int \dd t \; \normb{\abs{x}^{-\gamma} \, \ee^{- \ii t (-\Delta)^\gamma} \, \psi}^2 \;\leq\; c_{d, \gamma} \, \norm{\psi}^2\,,
\end{equation}
for some constant $c_{d, \gamma} > 0$.
\end{lemma}
\begin{remark}
The constant in \eqref{kato smoothing estimate} is sharp. Indeed, \eqref{kato smoothing estimate} is saturated if $\psi$ is Gaussian, as can be seen by explicit calculation.
\end{remark}
\begin{remark}
At the endpoint $\gamma = 1/2$ the dispersion law of the time evolution is $\omega(k) = \abs{k}$. Thus all spatial frequency components have the same propagation speed, i.e.\ there is no dispersion and the smoothing effect of the time evolution (which relies on the fast propagation of high spatial frequencies) vanishes. It is therefore not surprising that the endpoint $\gamma = 1/2$ is excluded in \eqref{generalized Kato smoothing}. Similarly, the claim is false at the other endpoint $\gamma = d/2$. This can be seen by noting that, for instance if $\psi$ is Gaussian, $\ee^{-\ii t (-\Delta)^{d/2}} \psi$ is nonzero in a neighbourhood of $0$ for small times. Since $\abs{x}^{-d}$ is not locally integrable, it follows that the left-hand side of \eqref{generalized Kato smoothing} is $\infty$.
\end{remark}
\begin{remark}
It is easy to see that our proof of \eqref{generalized Kato smoothing} remains valid if the power law potential $v(x) = \abs{x}^{-\gamma}$ is replaced with a potential $v$ satisfying
\begin{equation*}
\abs{\widehat{v^2}(k)} \;\lesssim\; \frac{1}{\abs{k}^{d - 2 \gamma}}\,,
\end{equation*}
where $\widehat{\,\cdot\,}$ denotes Fourier transformation.
\end{remark}

\begin{proof}[Proof of Lemma \ref{lemma: Kato smoothing}]

The left-hand side of \eqref{kato smoothing estimate} defines a quadratic form in $\psi$. By density, if we prove \eqref{kato smoothing estimate} for all $\psi \in \cal S$, it follows that \eqref{kato smoothing estimate} holds for all $\psi \in \LLL^2$. Let us therefore assume that $\psi \in \cal S$. By monotone convergence, we have
\begin{equation*}
\int \dd t \; \normb{\abs{x}^{-1} \, \ee^{\ii t \Delta} \, \psi}^2 \;=\; \lim_{\eta \downarrow 0} f(\eta)\,,
\end{equation*}
where
\begin{equation*}
f(\eta) \;\deq\;\int \dd t \; \normb{\abs{x}^{-1} \, \ee^{\ii t \Delta} \, \psi}^2 \, \ee^{-\frac{\eta}{2}  t^2} \;=\; 
\int \dd t \; \scalarb{\psi}{\ee^{- \ii t \Delta} \, \abs{x}^{-2} \, \ee^{\ii t \Delta} \, \psi} \, \ee^{-\frac{\eta}{2}  t^2}
\end{equation*}
In order to write the scalar product in Fourier space, we recall (see e.g.\ \cite{LiebLoss}) that, for $0 < \alpha < d$, we have
\begin{equation*}
\widehat{\abs{x}^{-\alpha}}(k) \;=\; 2^{d/2 - \alpha} \frac{\Gamma\pb{\frac{d - \alpha}{2}}}{\Gamma\pb{\frac{\alpha}{2}}} \frac{1}{\abs{k}^{d - \alpha}}\,.
\end{equation*}
In particular,
\begin{equation*}
\widehat{\abs{x}^{-2}}(k) \;=\; \frac{(2 \pi)^{d/2}}{(d-2) \abs{\SSS^{d-1}}} \frac{1}{\abs{k}^{d-2}}\,,
\end{equation*}
where $\SSS^{d-1}$ denotes the unit sphere in $\R^d$ and $\abs{\SSS^{d-1}}$ its surface measure.
Thus (writing $\psi$ instead of $\widehat{\psi}$) we get
\begin{equation*}
f(\eta) \;=\; \frac{(2 \pi)^{d/2}}{(2 \pi)^{d/2} (d - 2) \abs{\SSS^{d-1}}} \, \int \dd t \; \ee^{- \frac{\eta}{2} t^2} \, \int \dd p_1 \, \dd p_2 \; \ol{\psi(p_1)} \, \ee^{\ii t p_1^2} \, \frac{1}{\abs{p_1 - p_2}^{d - 2}} \, \ee^{- \ii t p_2^2} \, \psi(p_2)\,.
\end{equation*}
Using Fubini's theorem we get
\begin{align*}
f(\eta) &\;=\; \frac{1}{ (d - 2) \abs{\SSS^{d-1}}} \, \int \dd p_1 \, \dd p_2 \; \ol{\psi(p_1)}  \, \frac{1}{\abs{p_1 - p_2}^{d - 2}} \, \psi(p_2) \, \int \dd t \; \ee^{- \frac{\eta}{2} t^2} \, \ee^{\ii t (p_1^2 - p_2^2)}
\\
&\;=\; \frac{1}{ (d - 2) \abs{\SSS^{d-1}}} \, \int \dd p_1 \, \dd p_2 \; \ol{\psi(p_1)}  \, \frac{1}{\abs{p_1 - p_2}^{d - 2}} \, \psi(p_2) \, 2 \pi \frac{1}{\sqrt{ 2 \pi \eta}} \, \ee^{-\frac{1}{2 \eta} (p_1^2 - p_2^2)^2}
\\
&\;\leq\; \frac{2 \pi}{ (d - 2) \abs{\SSS^{d-1}}} \, \int \dd p_1 \, \dd p_2 \; \abs{\psi(p_1)} \, \abs{\psi(p_2)} \, \frac{1}{\abs{p_1 - p_2}^{d - 2}} \, \frac{1}{\sqrt{2 \pi \eta}} \, \ee^{-\frac{1}{2 \eta} (p_1^2 - p_2^2)^2}
\\
&\;\leq\; \frac{2 \pi}{ (d - 2) \abs{\SSS^{d-1}}} \, \int \dd p_1 \, \dd p_2 \; \abs{\psi(p_2)}^2 \, \frac{1}{\abs{p_1 - p_2}^{d - 2}} \, \frac{1}{\sqrt{2 \pi \eta}} \, \ee^{-\frac{1}{2 \eta} (p_1^2 - p_2^2)^2}\,,
\end{align*}
where in the last step we used the inequality $2 ab \leq a^2 + b^2$ and symmetry. This implies
\begin{equation*}
f(\eta) \;\leq\; \frac{2 \pi}{ (d - 2) \abs{\SSS^{d-1}}} \, \norm{\psi}^2 \, \sup_{p_2} \int \dd p_1 \; \frac{1}{\abs{p_1 - p_2}^{d - 2}} \, \frac{1}{\sqrt{2 \pi \eta}} \, \ee^{-\frac{1}{2 \eta} (p_1^2 - p_2^2)^2}\,.
\end{equation*}
Let us write $p_2 = \lambda p$ and $k \deq p_1 / \lambda$, where $\lambda > 0$ and $p \in \SSS^{d-1}$. Thus we get
\begin{equation*}
f(\eta) \;\leq\; \frac{2 \pi}{ (d - 2) \abs{\SSS^{d-1}}} \, \norm{\psi}^2 \, \sup_{\lambda > 0, \, p \in \SSS^{d-1}} \int \dd k \; \frac{1}{\abs{k - p}^{d - 2}} \, \frac{\lambda^2}{\sqrt{2 \pi \eta}} \, \ee^{-\frac{\lambda^4}{2 \eta} (k^2 - 1)^2}\,.
\end{equation*}
We do the integral over $k$ using polar coordinates: 
\begin{equation*}
k \;=\; \sqrt{v} \, e \,, \qquad \dd k \;=\; \frac{\sqrt{v}^{\,d-2}}{2} \, \dd v \, \dd e \,,\qquad v \in (0,\infty)\,,\, e \in \SSS^{d-1}\,,
\end{equation*}
where $\dd e$ denotes the usual surface measure on $\SSS^{d-1}$. This gives
\begin{equation*}
\int_0^\infty \dd v \; g(v) \, \frac{\lambda^2}{\sqrt{2 \pi \eta}} \, \ee^{-\frac{\lambda^4}{2 \eta} (v - 1)^2}\,,
\end{equation*}
where
\begin{equation} \label{definition of g(v)}
g(v) \;\deq\; \frac{\sqrt{v}^{\,d-2}}{2} \, \int_{\SSS^{d-1}} \dd e \; \frac{1}{\abs{\sqrt{v} \, e - p}^{d-2}} \;=\;
\frac{1}{2} \, \int_{\SSS^{d-1}} \dd e \; \frac{1}{\abs{e - p / \sqrt{v}}^{d-2}}
\,.
\end{equation}
Next, recall Newton's theorem for spherically symmetric mass distributions (see e.g.\ \cite{LiebLoss}): If $\mu$ is a spherically symmetric, finite, complex measure on $\R^d$, then
\begin{equation*}
\int \dd \mu(y) \; \frac{1}{\abs{x - y}^{d-2}} \;=\; \frac{1}{\abs{x}^{d - 2}} \int \dd \mu(y) \; \umat_{\{\abs{y} \leq \abs{x}\}} + \int \dd \mu(y) \; \frac{1}{\abs{y}^{d - 2}} \umat_{\{\abs{y} > \abs{x}\}}\,.
\end{equation*}
This yields
\begin{equation*}
g(v) \;=\;
\frac{1}{2} 
\begin{cases}
\absb{\SSS^{d-1}} \sqrt{v}^{\, d - 2} & \text{if } v \leq 1
\\
\absb{\SSS^{d-1}} & \text{if } v > 1 \,.
\end{cases}
\end{equation*}
Thus, $g$ is continuous and takes on its maximum value at $1$. Since
\begin{equation*}
\frac{\lambda^2}{\sqrt{2 \pi \eta}} \, \ee^{-\frac{\lambda^4}{2 \eta} (v - 1)^2}
\end{equation*}
is an approximate delta-function centred at $1$ it follows that
\begin{equation*}
\sup_{\lambda} \int_0^\infty \dd v \; g(v) \, \frac{\lambda^2}{\sqrt{2 \pi \eta}} \, \ee^{-\frac{\lambda^4}{2 \eta} (v - 1)^2} \;=\; \lim_{\lambda \to \infty} \int_0^\infty \dd v \; g(v) \, \frac{\lambda^2}{\sqrt{2 \pi \eta}} \, \ee^{-\frac{\lambda^4}{2 \eta} (v - 1)^2} \;=\; g(1)\,.
\end{equation*}
Thus,
\begin{equation*}
f(\eta) \;\leq\; \frac{1}{2} \, \frac{2 \pi}{ (d - 2) \abs{\SSS^{d-1}}} \, \absb{\SSS^{d-1}} \, \norm{\psi}^2 \;=\; \frac{\pi}{d - 2} \norm{\psi}^2\,.
\end{equation*}
This completes the proof of \eqref{kato smoothing estimate}.

The proof of \eqref{generalized Kato smoothing} follows the proof of \eqref{kato smoothing estimate} up to \eqref{definition of g(v)}. The claim then follows from
\begin{equation*}
\sup_{v \geq 0} \int_{\SSS^{d-1}} \dd e \; \frac{1}{\abs{e - p/ \sqrt{v}}^{d - 2 \gamma}} \;<\; \infty\,,
\end{equation*}
for $p \in \SSS^{d-1}$ and $2 \gamma > 1$.
\end{proof}

In order to avoid tedious discussions of operator domains in equations such as \eqref{Duhamel}, 
we introduce a cutoff to make the interaction potential bounded. For $\epsilon \geq 0$ set
\begin{equation*}
w^\epsilon(x) \;\deq\; w(x) I_{\{\abs{w(x)} \leq \epsilon^{-1}\}}\,,
\end{equation*}
so that $\norm{w^\epsilon}_\infty \leq \epsilon^{-1}$. Now \eqref{kato smoothing estimate} implies, for $d = 3$ and $\epsilon \geq 0$,
\begin{equation} \label{general Kato smoothing for cutoff potentials}
\int_{\R} \normb{w^\epsilon \, \me^{\mi t\Delta} \, \psi}^2 \; \md t
\;\leq\;
\int_{\R} \normb{w \, \me^{\mi t\Delta} \, \psi}^2 \; \md t
\;\leq\; \pi \kappa^2 \, \norm{\psi}^2\,.
\end{equation}
An immediate consequence is the following lemma.
\begin{lemma} \label{lemma: Kato smoothing for centre of mass}
Let $\Phi^{(n)} \in \mathcal{H}_\pm^{(n)}$. Then
\begin{equation} \label{kato smoothing for N-particles}
\int_\R \normb{W^\epsilon_{ij} \, \me^{-\mi t H_0} \, \Phi^{(n)}}^2 \; \md t \;\leq\; \frac{\pi \kappa^2}{2} \norm{\Phi^{(n)}}^2\,.
\end{equation}
\end{lemma}
\begin{proof}
By symmetry we may assume that $(i,j) = (1,2)$. Choose centre of mass coordinates $X \deq (x_1 + x_2)/2$ and $\xi = x_2 - x_1$, set $\tilde{\Phi}^{(n)}(X, \xi, x_3, \dots, x_n) \deq \Phi^{(n)}(x_1, \dots, x_n)$, and write
\begin{equation*}
\int_\R \normb{W^\epsilon_{12} \, \me^{-\mi t H_0} \, \Phi^{(n)}}^2 \; \md t \;=\; \int_\R \normb{w^\epsilon(\xi) \, \me^{2 \mi t \Delta_\xi} \, \tilde{\Phi}^{(n)}}^2 \; \md t\,,
\end{equation*}
since $H_0 = -\Delta_1 - \Delta_2 = -\Delta_X /2 - 2 \Delta_\xi$ and $[\Delta_X, w^\epsilon(\xi)] = 0$. Therefore, by \eqref{general Kato smoothing for cutoff potentials} and Fubini's theorem, we find
\begin{align*}
\int_\R \normb{W^\epsilon_{12} \, \me^{-\mi t H_0} \, \Phi^{(n)}}^2 \; \md t &\;=\; \int \md X \, \md x_3 \cdots \md x_n \int \md t  \, \md \xi \; \absb{w^\epsilon(\xi) \, \me^{2 \mi t \Delta_\xi} \, \tilde{\Phi}^{(n)}(X, \xi, x_3, \dots, x_n)}^2
\\
&\;\leq\; \frac{\pi \kappa^2}{2} \, \int \md X \, \md x_3 \cdots \md x_n \int \md \xi \; \absb{\tilde{\Phi}^{(n)}(X, \xi, x_3, \dots, x_n)}^2
\\
&\;=\; \frac{\pi \kappa^2}{2} \norm{\Phi^{(n)}}^2\,.
\qedhere
\end{align*}
\end{proof}
By Cauchy-Schwarz we then find that
\begin{equation} \label{kato smoothing estimate for l1 norm}
\int_0^t \normb{W^\epsilon_{ij,s} \, \Phi^{(n)}} \; \md s \;\leq\; t^{1/2} \, \pbb{\int_\R \normb{W^\epsilon_{ij} \, \me^{-\mi s H_0} \, \Phi^{(n)}}^2 \md s}^{1/2} \;\leq\; \pbb{\frac{\pi \kappa^2 t}{2}}^{1/2} \norm{\Phi^{(n)}}\,.
\end{equation}
By iteration, this implies that, for all elementary terms $\alpha$,
\begin{equation} \label{iterated kato smoothing estimate}
\int_0^t \md t_1 \dots \int_0^t \md t_{k} \;
\normb{G^{(k,l)(\alpha), \epsilon}_{t, \ul{t}}(a^{(p)}) \Phi^{(p+k-l)}}
\leq\; 
\pbb{\frac{\pi \kappa^2 t}{2}}^{k/2}\, \norm{a^{(p)}}\, \norm{\Phi^{(p+k-l)}}\,,
\end{equation}
where the superscript $\epsilon$ reminds us that $G^{(k,l)(\alpha), \epsilon}_{t, \ul{t}}(a^{(p)})$ is computed with the regularized potential $w^\epsilon$.
Thus one finds
\begin{equation*}
\normb{G^{(k,l), \epsilon}_t(a^{(p)})} \;\leq\; c(p,k,l) \, \pbb{\frac{\pi \kappa^2 t}{2}}^{k/2}\,\norm{a^{(p)}}\,,
\end{equation*}
for all $\epsilon \geq 0$.

Unfortunately, the above procedure does not recover the factor $1/k!$ arising from the time-integration over the $k$-simplex $\Delta^{k}(t)$, which is essential for our convergence estimates. First iterating \eqref{kato smoothing for N-particles} and then using Cauchy-Schwarz yields a factor $1/\sqrt{k!}$, which is still not good enough.

A solution to this problem must circumvent the highly wasteful procedure of replacing the integral over $\Delta^{k}(t)$ with an integral over $[0,t]^k$. The key observation is that, in the sum over all labelled diagrams, each diagram appears of the order of $k!$ times with different labellings.

\subsection{Graph counting}
In order to make the above idea precise, we make use of graphs (related to the above diagrams) to index terms in our expansion of the multiple commutator
\begin{equation} \label{multiple commutator for graphs}
\frac{(\mi N)^k}{2^k} \qB{\widehat{\A}_N(W_{t_k}), \dots \qB{\widehat{\A}_N(W_{t_1}), \widehat{\A}_N(a^{(p)}_t)}\dots}\,.
\end{equation}
The idea is to assign to each second quantized operator a vertex $v = 0, \dots, k$, and to represent each creation and annihilation with an incident edge. A pairing of an annihilation operator with a creation operator is represented by joining the corresponding edges. The vertex $0$ has $2p$ edges and the vertices $1, \dots, k$ have $4$ edges. We call the vertex $0$ the \emph{root}.

The edges incident to each vertex $v$ are labelled using a pair $\lambda = (d, i)$, where $d = a,c$ is the \emph{direction} ($a$ stands for ``annihilation'' and $c$ for ``creation'') and $i$ labels edges of the same direction; $i = 1, \dots ,p$ if $v = 0$ and $i = 1,2$ if $v = 1, \dots, k$. Thus, a labelled edge is of the form $\{(v_1, \lambda_1), (v_2, \lambda_2)\}$. Graphs $\mathcal{G}$ with such labelled edges are graphs over the vertex set $V(\mathcal{G}) = \{(v, \lambda)\}$. We denote the set of edges of a graph $\mathcal{G}$ (a set of unordered pairs of vertices in $V(\mathcal{G})$) by $E(\mathcal{G})$. The degree of each $(v, \lambda)$ is either 0 or 1; we call $(v, \lambda)$ an \emph{empty edge} of $v$ if its degree is 0. We often speak of connecting two empty edges, as well as removing a nonempty edge; the definitions are self-explanatory.

We may drop the edge labelling of $\mathcal{G}$ to obtain a (multi)graph $\widetilde{\mathcal{G}}$ over the vertex set $\{0, \dots, k\}$: Each edge $\{(v_1, \lambda_1), (v_2,\lambda_2)\} \in E(\mathcal{G})$ gives rise to the edge $\{v_1, v_2\} \in E(\widetilde{\mathcal{G}})$.
We understand a path in $\mathcal{G}$ to be a sequence of edges in $E(\mathcal{G})$ such that two consecutive edges are adjacent in the graph $\widetilde{\mathcal{G}}$. This leads to the notions of connectedness of $\mathcal{G}$ and loops in $\mathcal{G}$.

The \emph{admissible graphs} -- i.e.\ graphs indexing a choice of pairings in the multiple commutator \eqref{multiple commutator for graphs} -- are generated by the following ``growth process''. We start with the empty graph $\mathcal{G}_0$, i.e.\ $E(\mathcal{G}_0) = \emptyset$. In a first step, we choose one or two empty edges of $1$ of the same direction and connect each of them to an empty edge of $0$ of opposite direction. Next, we choose one or two empty edges of $2$ of the same direction and connect each of them to an empty edge of $0$ or $1$ of opposite direction. We continue in this manner for all vertices $3, \dots, k$. We summarize some key properties of admissible graphs $\mathcal{G}$.
\begin{itemize}
\item[(a)]
$\mathcal{G}$ is connected.
\item[(b)]
The degree of each $(v, \lambda)$ is either 0 or 1
\item[(c)]
The labelled edge $\{(v_1, \lambda_1), (v_2, \lambda_2)\} \in E(\mathcal{G})$ only if $\lambda_1$ and $\lambda_2$ have opposite directions.
\end{itemize}
Property (c) implies that each graph $\mathcal{G}$ has a canonical directed representative, where each edge is ordered from the $a$-label to the $c$-label.
\begin{figure}[ht!]
\psfrag{0}[][]{$0$}
\psfrag{1}[][]{$1$}
\psfrag{2}[][]{$2$}
\psfrag{3}[][]{$3$}
\psfrag{4}[][]{$4$}
\psfrag{5}[][]{$5$}
\psfrag{6}[][]{$6$}
\psfrag{7}[][]{$7$}
\psfrag{a1}[][]{$\scriptstyle{a,1}$}
\psfrag{a2}[][]{$\scriptstyle{a,2}$}
\psfrag{a3}[][]{$\scriptstyle{a,3}$}
\psfrag{a4}[][]{$\scriptstyle{a,4}$}
\psfrag{c1}[][]{$\scriptstyle{c,1}$}
\psfrag{c2}[][]{$\scriptstyle{c,2}$}
\psfrag{c3}[][]{$\scriptstyle{c,3}$}
\psfrag{c4}[][]{$\scriptstyle{c,4}$}
\vspace{0.5cm}
\begin{center}
\includegraphics[width = 12cm]{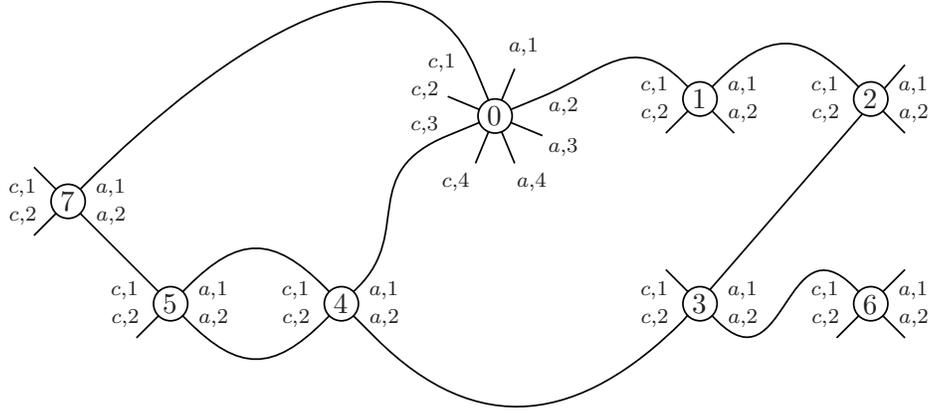}
\end{center}
\caption{\textit{An admissible graph of type $(p = 4, k = 7, l = 3)$.} \label{figure: full graph}}
\end{figure}
See Figure \ref{figure: full graph} for an example of such a graph.

We call a graph $\mathcal{G}$ of type $(p,k,l)$ whenever it is admissible and it contains $l$ loops. We denote by $\mathscr{G}(p,k,l)$ the set of graphs of type $(p,k,l)$.

By definition of admissible graphs, each contraction in \eqref{multiple commutator for graphs} corresponds to a unique admissible graph. A contraction consists of at least $k$ and at most $2k$ pairings. A contraction giving rise to a graph of type $(p,k,l)$ has $k+l$ pairings. The summand in \eqref{multiple commutator for graphs} corresponding to any given $l$-loop contraction is given by an \emph{elementary term} of the form
\begin{equation} \label{second quantized elementary term}
\frac{(\mi N)^k}{2^k \, N^{k+l}} \widehat{\A}_N \pb{b^{(p+k-l)}}\,,
\end{equation}
where the $(p+k-l)$-particle operator $b^{(p+k-l)}$ is of the form
\begin{equation} \label{elementary term for contractions}
b^{(p+k-l)} \;=\; P_\pm \, W_{i_1 j_1, t_{v_1}} \cdots W_{i_r j_r, t_{v_r}} \,\pb{a_t^{(p)} \otimes \umat^{(k-l)}} \, W_{i_{r+1} j_{r+1}, t_{v_{r+1}}} \cdots W_{i_k j_k, t_{v_k}} P_\pm \,,
\end{equation}
for some $r = 0, \dots, k$. Indeed, the (anti)commutation relations \eqref{anticommutation relations} imply that each pairing produces a factor of $1/N$. Furthermore, the creation and annihilation operators of each summand corresponding to any given contraction are (by definition) Wick ordered, and one readily sees that the associated integral kernel corresponds to an operator of the form \eqref{elementary term for contractions}. Thus we recover the splitting \eqref{splitting of multiple commutator}, whereby $G^{(k,l)}_{t, t_1, \dots, t_k}(a^{(p)})$ is a sum, indexed by all $l$-loop graphs, of elementary terms of the form \eqref{elementary term for contractions}.

As remarked above, we need to exploit the fact that many graphs have the same topological structure, i.e.\ can be identified after some permutation of the labels $\{1, \dots, k\}$ of the vertices corresponding to interaction operators. We therefore define an equivalence relation on the set of graphs: $\mathcal{G} \sim \mathcal{G}'$ if and only if there exists a permutation $\sigma \in S_k$ such that $\mathcal{G}' = R_\sigma(\mathcal{G})$. Here $R_\sigma(\mathcal{G})$ is the graph defined by
\begin{equation*}
\{(v_1, \lambda_1) , (v_2, \lambda_2)\} \in E(R_\sigma(\mathcal{G})) \;\Longleftrightarrow\; \{(\sigma(v_1), \lambda_1), (\sigma(v_2), \lambda_2)\} \in E(\mathcal{G})\,,
\end{equation*}
where $\sigma(0) \equiv 0$. 
We call equivalence classes $[\mathcal{G}]$ \emph{graph structures}, and denote the set of graph structures of admissible graphs of type $(p,k,l)$ by $\mathscr{Q}(p,k,l)$. 

Note that, in general, $R_\sigma(\mathcal{G})$ need not be admissible if $\mathcal{G}$ is admissible. It is convenient to increase $\mathscr{G}(p,k,l)$ to include all $R_\sigma(\mathcal{G})$ where $\sigma \in S_k$ and $\mathcal{G}$ is admissible. In order to keep track of the admissible graphs in this larger set, we introduce the symbol $i_\mathcal{G}$ which is by definition 1 if $\mathcal{G} \in \mathscr{G}(p,k,l)$ is admissible and 0 otherwise. Because $R_\sigma(\mathcal{G}) \neq \mathcal{G}$ if $\sigma \neq \id$,
\begin{equation} \label{size of set of graph structures}
\absb{\mathscr{G}(p,k,l)} \;=\; k! \, \absb{\mathscr{Q}(p,k,l)}\,.
\end{equation}

Our goal is to find an upper bound on the number of graph structures of type $(p,k,l)$, which is sharp enough to show convergence of the Schwinger-Dyson series \eqref{Schwinger-Dyson series}. Let us start with tree graphs: $l = 0$.
In this case the number of graph structures is equal to $2^k$ times the number of ordered trees\footnote{An ordered tree is a rooted tree in which the children of each vertex are ordered.} with $k+1$ vertices, whose root has at most $2p$ children and whose other vertices have at most $3$ children. The factor $2^k$ arises from the fact that each vertex $v = 1, \dots, k$ can use either of the two empty edges of compatible direction to connect to its parent. We thus need some basic facts about ordered trees, which are covered in the following (more or less standard) combinatorial digression.

For $x,t \in \R$ and $n \in \N$ define
\begin{equation}
A_n(x,t) \;\deq\; \frac{x}{x + nt} \binom{x + nt}{n}
\end{equation}
as well as $A_0(x,t) \deq 1$.
After some juggling with binomial coefficients one finds
\begin{equation}
\sum_{k = 0}^n A_k(x,t) A_{n - k}(y,t) \;=\; A_n(x+y,t)\,;
\end{equation}
see \cite{Knuth1998} for details. Therefore
\begin{equation}
\sum_{n_1 + \cdots + n_r = n} A_{n_1}(x_1, t) \cdots A_{n_r}(x_r, t) \;=\; A_n(x_1 + \cdots + x_r, t)\,.
\end{equation}
Set 
\begin{equation} \label{definition of Catalan numbers}
C_n^m \;\deq\; A_n(1,m) \;=\; \frac{1}{1 + nm} \binom{1 + nm}{n} \;=\; \frac{1}{n(m - 1) + 1} \binom{nm}{n}\,,
\end{equation}
the $n$'th $m$-ary \emph{Catalan number}.
Thus we have
\begin{equation} \label{sum over Catalan numbers}
\sum_{n_1 + \cdots + n_r = n} C_{n_1}^m \cdots C_{n_r}^m \;=\; \frac{r}{r + nm} \binom{r + nm}{n}\,.
\end{equation}
In particular,
\begin{equation} \label{recursive relation for Catalan numbers}
\sum_{n_1 + \cdots + n_m = n-1} C_{n_1}^m \cdots C_{n_m}^m \;=\; C_n^m\,.
\end{equation}
Define an $m$-tree to be an ordered tree such that each vertex has at most $m$ children.
The number of $m$-trees with $n$ vertices is equal to $C_n^m$. This follows immediately from $C^m_0 = 1$ and from \eqref{recursive relation for Catalan numbers}, which expresses that all trees of order $n$ are obtained by adding $m$ (possibly empty) subtrees of combined order $n-1$ to the root.

We may now compute $\abs{\mathscr{Q}(p,k,0)}$. Since the root of the tree has at most $2p$ children, we may express $\abs{\mathscr{Q}(p,k,0)}$ as the number of ordered forests comprising $2p$ (possibly empty) 3-trees whose combined order is equal to $k$. Therefore, by \eqref{sum over Catalan numbers},
\begin{equation} \label{bound on the number of tree graphs}
\abs{\mathscr{Q}(p,k,0)} \;=\; 2^k \sum_{n_1 + \cdots + n_{2p} = k} C^3_{n_1} \cdots C^3_{n_{2p}} \;=\; 2^k \frac{2p}{2p + 3k} \binom{2p + 3k}{k}\,.
\end{equation}

Next, we extend this result to all values of $l$ in the form of an upper bound on $\abs{\mathscr{Q}(p,k,l)}$.
\begin{lemma} \label{lemma: bound on l-loop graphs}
Let $p,k, l \in \N$. Then
\begin{equation}
\abs{\mathscr{Q}(p,k,l)} \;\leq\; 2^k \binom{k}{l} \, \binom{2p + 3k}{k} \, (p+k-l)^l\,.
\end{equation}
\end{lemma}
\begin{proof}
The idea is to remove edges from $\mathcal{G} \in \mathscr{G}(p,k,l)$ to obtain a tree graph, and then use the special case \eqref{bound on the number of tree graphs}. 

In addition to the properties (a) -- (c) above, we need the following property of $\mathscr{G}(p,k,l)$:
\begin{itemize}
\item[(d)]
If $\mathcal{G} \in \mathscr{G}(p,k,l)$ then there exists a subset $\mathcal{V} \subset \{1, \dots, k\}$ of size $l$ and a choice of direction $\delta : \mathcal{V} \to \{a,c\}$ such that, for each $v \in \mathcal{V}$, both edges of $v$ with direction $\delta(v)$ are nonempty. Denote by $\mathcal{E}(v) \subset E(\mathcal{G})$ the set consisting of the two above edges. We additionally require that removing one of the two edges of $\mathcal{E}(v)$ from $\mathcal{G}$, for each $v \in \mathcal{V}$, yields a tree graph, with the property that, for each $v \in \mathcal{V}$, the remaining edge of $\mathcal{E}(v)$ is contained in the unique path connecting $v$ to the root.
\end{itemize}
This is an immediate consequence of the growth process for admissible graphs. The set $\mathcal{V}$ corresponds to the set of vertices whose addition produces two edges. Note that property (d) is independent of the representative and consequently holds also for non-admissible $\mathcal{G} \in \mathscr{G}(p,k,l)$. 

Before coming to our main argument, we note that a tree graph $\mathcal{T} \in \mathscr{G}(p,k,0)$ gives rise to a natural lexicographical order on the vertex set $\{1, \dots, k\}$. Let $v \in \{1, \dots, k\}$. There is a unique path that connects $v$ to the root. Denote by $0 = v_1, v_2, \dots, v_q = v$ the sequence of vertices along this path. For each $j = 1, \dots, q - 1$, let $\lambda_j$ be the label of the edge $\{v_j, v_{j+1}\}$ at $v_j$. We assign to $v$ the string $S(v) \deq (\lambda_1, \dots, \lambda_{q - 1})$. Choose some (fixed) ordering of the sets of labels $\{\lambda\}$, for each $v$. Then the set of vertices $\{1, \dots, k\}$ is ordered according to the lexicographical order of the string $S(v)$.

We now start removing loops from a given graph $\mathcal{G} \in \mathscr{G}(p,k,l)$. Define $\mathcal{R}_1$ as the graph obtained from $\mathcal{G}$ by removing all edges in $\bigcup_{v \in \mathcal{V}} \mathcal{E}(v)$. By property (d) above, $\mathcal{R}_1$ is a forest comprising $l$ trees. Define $\mathcal{T}_1$ as the connected component of $\mathcal{R}_1$ containing the root. Now we claim that there is at least one $v \in \mathcal{V}$ such that both edges of $\mathcal{E}(v)$ are incident to a vertex of $\mathcal{T}_1$. Indeed, were this not the case, we could choose for each $v \in \mathcal{V}$ an edge in $\mathcal{E}(v)$ that is not incident to any vertex of $\mathcal{T}_1$. Call $\mathcal{R}'_1$ the graph obtained by adding all such edges to $\mathcal{R}_1$. Now, since no vertex in $\mathcal{V}$ is in the connected component of $\mathcal{R}_1$, it follows that no vertex in $\mathcal{V}$ is in the connected component $\mathcal{R}'_1$. This is a contradiction to property (d) which requires that $\mathcal{R}'_1$ should be a (connected) tree.

Let us therefore consider the set $\tilde{\mathcal{V}}$ of all $v \in \mathcal{V}$ such that both edges of $\mathcal{E}(v)$ are incident to a vertex of $\mathcal{T}_1$. We have shown that $\tilde{\mathcal{V}} \neq \emptyset$.
For each choice of $v$ and $e$, where $v \in \tilde{\mathcal{V}}$ and $e \in \mathcal{E}(v)$, we get a forest of $l-1$ trees by adding $e$ to the edge set of $\mathcal{R}_1$. Then $v$ is in the same tree as the root, so that each such choice of $v$ and $e$ yields a string $S(v)$ as described above. We choose $v_1$ and $e(v_1)$ as the unique couple that yields the smallest string (note that different choices have different strings). Finally, set $\mathcal{G}_1$ equal to $\mathcal{G}$ from which $e(v_1)$ has been removed, and $\mathcal{V}_1 \deq \mathcal{V} \setminus \{v\}$.

We have thus obtained an $(l-1)$-loop graph $\mathcal{G}_1$ and a set $\mathcal{V}_1$ of size $l-1$, which together satisfy the property (d). We may therefore repeat the above procedure. In this manner we obtain the sequences $v_1, \dots, v_l$ and $\mathcal{G}_1, \dots, \mathcal{G}_l$. Note that $\mathcal{G}_l$ is obtained by removing the edges $e(v_1), \dots, e(v_l)$ from $\mathcal{G}$, and is consequently a tree graph. Also, by construction, the sequence $v_1, \dots, v_l$ is increasing in the lexicographical order of $\mathcal{G}_l$.

Next, consider the tree graph $\mathcal{G}_l$. Each edge $e(v_j)$ connects the single empty edge of $v_j$ with direction $\delta(v_j)$ with an empty edge of opposite direction of a vertex $v$, where $v$ is smaller than $v_j$ in the lexicographical order of $\mathcal{G}_l$. It is easy to see that, for each $j$, there are at most $(p+k-l)$ such connections.

We have thus shown that we can obtain any $\mathcal{G} \in \mathscr{G}(p,k,l)$ by choosing some tree $\mathcal{G}_l \in \mathscr{G}(p,k,0)$, choosing $l$ elements $v_j$ out of $\{1, \dots, k\}$, ordering them lexicographically (according to the order of $\mathcal{G}_l$) and choosing an edge out of at most $(p+k-l)$ possibilities for $v_1, \dots, v_l$. Thus,
\begin{equation*}
\absb{\mathscr{G}(p,k,l)} \;\leq\; \binom{k}{l} (p+k-l)^l \, \absb{\mathscr{G}(p,k,0)}\,.
\end{equation*}
The claim then follows from \eqref{size of set of graph structures} and \eqref{bound on the number of tree graphs}.
\end{proof}

\subsection{Proof of convergence}
We are now armed with everything we need in order to estimate $\int_{\Delta^k(t)} \md \ul{t} \; G^{(k,l)}_{t,\ul{t}}(a^{(p)})$. Recall that
\begin{equation} \label{expansion using graphs}
G^{(k, l)}_{t,t_1,\dots,t_k}(a^{(p)}) \;=\; \frac{\mi^k}{2^k} \sum_{\mathcal{G} \in \mathscr{G}(p,k,l)} i_{\mathcal{G}} \, G^{(k, l)(\mathcal{G})}_{t,t_1,\dots, t_k}(a^{(p)})\,,
\end{equation}
where $G^{(k, l)(\mathcal{G})}_{t,t_1,\dots, t_k}(a^{(p)})$ is an elementary term of the form \eqref{elementary term for contractions} indexed by the graph $\mathcal{G}$. We rewrite this using graph structures. Pick some choice $\mathcal{P} : \mathscr{Q}(p,k,l) \to \mathscr{G}(p,k,l)$ of representatives. Then we get
\begin{align*}
G^{(k, l)}_{t,t_1,\dots,t_k}(a^{(p)}) &\;=\; \frac{\mi^k}{2^k} \sum_{\mathcal{Q} \in \mathscr{Q}(p,k,l)} \sum_{\mathcal{G} \in \mathcal{Q}} i_{\mathcal{G}} \, G^{(k, l)(\mathcal{G})}_{t,t_1,\dots, t_k}(a^{(p)})
\\
&\;=\; \frac{\mi^k}{2^k} \sum_{\mathcal{Q} \in \mathscr{Q}(p,k,l)} \sum_{\sigma \in S_k} i_{R_\sigma(\mathcal{P}(\mathcal{Q}))} \, G^{(k, l)(R_\sigma(\mathcal{P}(\mathcal{Q})))}_{t,t_1,\dots, t_k}(a^{(p)})
\,.
\end{align*}
Now, by definition of $R_\sigma$, we see that
\begin{equation*}
G^{(k, l)(R_\sigma(\mathcal{G}))}_{t,t_1,\dots, t_k}(a^{(p)}) \;=\; G^{(k, l)(\mathcal{G})}_{t,t_{\sigma(1)},\dots, t_{\sigma(k)}}(a^{(p)})\,.
\end{equation*}
Thus,
\begin{align*}
\int_{\Delta^k(t)} \md \ul{t} \; G^{(k, l)}_{t,t_1,\dots,t_k}(a^{(p)}) &\;=\; \frac{\mi^k}{2^k} \sum_{\mathcal{Q} \in \mathscr{Q}(p,k,l)} \sum_{\sigma \in S_k} i_{R_\sigma(\mathcal{P}(\mathcal{Q}))} \int_{\Delta^k(t)} \md \ul{t} \; G^{(k, l)(\mathcal{P}(\mathcal{Q}))}_{t,t_{\sigma(1)},\dots, t_{\sigma(k)}}(a^{(p)})
\\
&\;=\; \frac{\mi^k}{2^k} \sum_{\mathcal{Q} \in \mathscr{Q}(p,k,l)} \int_{\Delta^k_{\mathcal{Q}}(t)} \md \ul{t} \; G^{(k, l)(\mathcal{P}(\mathcal{Q}))}_{t,t_1,\dots, t_k}(a^{(p)})\,,
\end{align*}
where
\begin{equation*}
\Delta^k_{\mathcal{Q}}(t) \;\deq\; \{(t_1, \dots, t_k) \,:\, \exists \sigma \in S_k \,:\, i_{R_\sigma(\mathcal{P}(\mathcal{Q}))} = 1 ,\, (t_{\sigma(1)}, \dots, t_{\sigma(k)}) \in \Delta^k(t)\} \;\subset\; [0,t]^k
\end{equation*}
is a union of disjoint simplices.

Therefore, \eqref{kato smoothing estimate for l1 norm} and \eqref{elementary term for contractions} imply, for any $\Phi^{(p+k-l)} \in \mathcal{H}^{(p+k-l)}_\pm$, that
\begin{align*}
\normbb{\int_{\Delta^k(t)} \md \ul{t} \; G^{(k, l)}_{t,\ul{t}}(a^{(p)}) \, \Phi^{(p+k-l)}}
&\;\leq\;
\frac{1}{2^k} \sum_{\mathcal{Q} \in \mathscr{Q}(p,k,l)} \int_{\Delta^k_{\mathcal{Q}}(t)} \md \ul{t} \; \normb{G^{(k, l)(\mathcal{P}(\mathcal{Q}))}_{t,t_1,\dots, t_k}(a^{(p)}) \, \Phi^{(p+k-l)}}
\\
&\;\leq\;
\frac{1}{2^k} \sum_{\mathcal{Q} \in \mathscr{Q}(p,k,l)} \int_{[0,t]^k} \md \ul{t} \; \normb{G^{(k, l)(\mathcal{P}(\mathcal{Q}))}_{t,t_1,\dots, t_k}(a^{(p)}) \, \Phi^{(p+k-l)}}
\\
&\;\leq\;
\frac{1}{2^k} \sum_{\mathcal{Q} \in \mathscr{Q}(p,k,l)} \pbb{\frac{\pi \kappa^2 t}{2}}^{k/2} \norm{a^{(p)}} \norm{\Phi^{(p+k-l)}}
\\
&\;\leq\;
\binom{2p + 3k}{k} \, \binom{k}{l} \, (p+k-l)^l \pbb{\frac{\pi \kappa^2 t}{2}}^{k/2} \norm{a^{(p)}} \norm{\Phi^{(p+k-l)}}\,, 
\end{align*}
where the last inequality follows from Lemma \ref{lemma: bound on l-loop graphs}.
Of course, the above treatment remains valid for regularized potentials. We summarize:
\begin{equation} \label{fundamental estimate}
\normb{G^{(k,l), \epsilon}_t(a^{(p)})} \;\leq\;
\binom{2p+3k}{k} \binom{k}{l} (p + k - l)^l \pbb{\frac{\pi \kappa^2 t}{2}}^{k/2}\, \norm{a^{(p)}}\,,
\end{equation}
for all $\epsilon \geq 0$.

Using \eqref{fundamental estimate} we may now proceed exactly as in the case of a bounded interaction potential.
Let
\begin{equation}
\rho(\kappa, \nu) \;\deq\; \frac{1}{128 \pi \kappa^2 \nu^2}\,.
\end{equation}
The removal of the cutoff and summary of the results are contained in
\begin{lemma} \label{lemma: Schwinger-Dyson expansion}
Let $t < \rho(\kappa, \nu)$. Then we have on $\mathcal{H}^{(\nu N)}_\pm$
\begin{equation} \label{final form of Schwinger-Dyson expansion}
\me^{\mi t H_N} \, \widehat{\A}_N(a^{(p)}) \, \me^{- \mi t H_N} \;=\; \sum_{k = 0}^\infty \sum_{l = 0}^k \frac{1}{N^l} \, \widehat{\A}_N\pb{G^{(k,l)}_t(a^{(p)})}\,,
\end{equation}
in operator norm, uniformly in $N$. Furthermore, for $L \in \N$, we have the $1/N$-expansion
\begin{equation} \label{1/N expansion}
\me^{\mi t H_N} \, \widehat{\A}_N(a^{(p)}) \, \me^{- \mi t H_N} \;=\; \sum_{l = 0}^{L-1} \frac{1}{N^l} \sum_{k = l}^\infty \, \widehat{\A}_N\pb{G^{(k,l)}_t(a^{(p)})} + O \pbb{\frac{1}{N^L}}\,,
\end{equation}
where the sum converges on $\mathcal{H}^{(\nu N)}_\pm$ uniformly in $N$.
\end{lemma}
\begin{proof}
Using \eqref{fundamental estimate} we may repeat the proof of Lemma \ref{convergence for bounded potential} to the letter to prove the statements about convergence. Thus \eqref{final form of Schwinger-Dyson expansion} holds for all $\epsilon > 0$. The proof of \eqref{final form of Schwinger-Dyson expansion} for $\epsilon = 0$ follows by approximation and is banished to Appendix \ref{section: removal of cutoff}.
\end{proof}

\section{The Mean-Field Limit} \label{section: mean-field limit for bosons}
In this section we identify the mean-field dynamics as the dynamics given by the Hartree equation.
\subsection{The Hartree equation} \label{section: Hartree}
The Hartree equation reads
\begin{equation} \label{Hartree equation}
\mi \partial_t \psi \;=\; h \psi + (w * \abs{\psi}^2) \psi\,.
\end{equation}
It is the equation of motion of a classical Hamiltonian system with phase space $\Gamma \deq \HHH^1(\R^3)$. Here $\HHH^1(\R^3)$ is the usual Sobolev space of index one. 
In analogy to $\widehat{\A}_N$ we define $\A$ as the map from closed operators on $\mathcal{H}^{(p)}_+$ to functions on phase space, through 
\begin{align*}
\A(a^{(p)})(\psi) &\;\deq\; \scalar{\psi^{\otimes p}}{a^{(p)} \, \psi^{\otimes p}}
\\
&\;=\; \int \md x_1 \cdots \md x_p \, \md y_1 \cdots \md y_p \;
\bar{\psi}(x_p) \cdots \bar{\psi}(x_1) \, a^{(p)}(x_1, \dots, x_p;y_1, \dots, y_p) \, 
\psi(y_1) \cdots \psi(y_p)\,.
\end{align*}
We define the space of ``observables'' $\mathfrak{A}$ as the linear hull of $\h{A(a^{(p)}) \,:\, p \in \N,\, a^{(p)} \in \mathcal{B}(\mathcal{H}^{(p)}_+)}$.

The Hamilton function is given by
\begin{equation*}
H \;\deq\; \A(h) + \frac{1}{2} \A(W)\,,
\end{equation*}
i.e.
\begin{equation} \label{classical Hamiltonian}
H(\psi) \;=\; \int \md x \; \abs{\nabla \psi}^2 + \frac{1}{2} \int \md x \; (w * \abs{\psi}^2) \abs{\psi}^2
\;=\; \scalar{\psi}{h \,\psi} + \frac{1}{2} \scalar{\psi^{\otimes 2}}{W \, \psi^{\otimes 2}}\,.
\end{equation}
Using the Hardy-Littlewood-Sobolev and Sobolev inequalities (see e.g.\ \cite{LiebLoss}) one sees that $H(\psi)$ is well-defined on $\Gamma$:
\begin{equation*}
\int \dd x \, \dd y \; \frac{\abs{\psi(x)}^2 \, \abs{\psi(y)}^2}{\abs{x - y}} \;\lesssim\; \normb{\abs{\psi}^2}_{6/5}^2 \;=\; \norm{\psi}_{12/5}^4 \;\lesssim\; \norm{\psi}^4_{\HHH^1}\,,
\end{equation*}
where the symbol $\lesssim$ means the left side is bounded by the right side multiplied by a positive constant that is independent of $\psi$.

The Hartree equation is equivalent to
\begin{equation*}
\mi \partial_t \psi \;=\; \partial_{\bar{\psi}} H(\psi)\,.
\end{equation*}
The symplectic form on $\Gamma$ is given by
\begin{equation*}
\omega = \mi \int \md x \; \md \bar{\psi}(x) \wedge \md \psi(x)\,,
\end{equation*}
which induces a Poisson bracket given by 
\begin{equation*}
\{\psi(x), \bar{\psi}(y)\} \;=\; \mi \delta(x-y) \,, \qquad \{\psi(x), \psi(y)\} \;=\; \{\bar{\psi}(x), \bar{\psi}(y)\} \;=\; 0 \,.
\end{equation*}
For $A, B \in \mathfrak{A}$ we have that
\begin{equation*}
\{A, B\}(\psi) \;=\; \mi \int \md x\; \qb{\partial_\psi A(\psi) \, \partial_{\bar{\psi}} B(\psi) - \partial_\psi B(\psi) \, \partial_{\bar{\psi}} A(\psi)}\,.
\end{equation*}

The ``mass'' function
\begin{equation*}
N(\psi) \;\deq\; \int \md x \; \abs{\psi}^2
\end{equation*}
is the generator of the gauge transformations $\psi \mapsto \me^{-\mi \theta} \psi$. By the gauge invariance of the Hamiltonian, $\{H, N\} = 0$, we conclude, at least formally, that $N$ is a conserved quantity. Similarly, the energy $H$ is formally conserved.

The space of observables $\mathfrak{A}$ has the following properties.
\begin{itemize}
\item[(i)]
$\ol{\A(a^{(p)})} = \A\pb{(a^{(p)})^*}$.
\item[(ii)]
If $a^{(p)} \in \mathcal{B}(\mathcal{H}^{(p)}_+)$ and $b \in \mathcal{B}(\mathcal{H})$, then 
\begin{equation*}
\A(a^{(p)})(b \psi) \;=\; \A\pb{(b^*)^{\otimes p} a^{(p)} b^{\otimes p}}(\psi)\,.
\end{equation*}
\item[(iii)]
If $a^{(p)}$ and $b^{(q)}$ are $p$- and $q$-particle operators, respectively, then
\begin{equation} \label{Poisson bracket computed}
\hb{\A(a^{(p)}), \A(b^{(q)})} \;=\; \mi pq \A \pb{\qb{a^{(p)}, b^{(q)}}_1}\,.
\end{equation}
\item[(iv)]
If $a^{(p)} \in \mathcal{B}(\mathcal{H}^{(p)}_+)$, then
\begin{equation} \label{bound on classical observable}
\abs{\A(a^{(p)})(\psi)} \;\leq\; \norm{a^{(p)}} \, \norm{\psi}^{2p}\,.
\end{equation}
\end{itemize}
The free time evolution
\begin{equation*}
\phi^t_0(\psi) \;\deq\; \me^{-\mi th} \psi
\end{equation*}
is the Hamiltonian flow corresponding to the free Hamilton function $\A(h)$. We abbreviate the free time evolution of observables $A \in \mathfrak{A}$ by $A_t \;\deq\; A \circ \phi^t_0$. Thus, $\A(a^{(p)})_t = \A(a^{(p)}_t)$.
%W(\psi) \;\deq\; \frac{1}{2} \int dx \; (w * \abs{\psi}^2) \abs{\psi}^2\,.

In order to define the Hamiltonian flow on all of $\LLL^2(\R^3)$, we rewrite the Hartree equation \eqref{Hartree equation} with initial data $\psi(0) = \psi$ as an integral equation
\begin{equation} \label{integral Hartree}
\psi(t) \;=\; \me^{-\mi th} \psi - \mi \int_0^t \md s\; \me^{-\mi (t-s) h} (w * \abs{\psi(s)}^2) \psi(s)\,.
\end{equation}

\begin{lemma} \label{lemma: Hartree wellposedness}
Let $\psi \in \LLL^2(\R^3)$. Then \eqref{integral Hartree} has a unique global solution $\psi(\cdot) \in \CCC(\R; \LLL^2(\R^3))$, which depends continuously on the initial data $\psi$. Furthermore, $\norm{\psi(t)} = \norm{\psi}$ for all $t$. Finally, we have a Schwinger-Dyson expansion for observables: Let $a^{(p)} \in \mathcal{B}(\mathcal{H}_+^{(p)})$, $\nu > 0$ and $t < \rho(\kappa, \nu)$. Then
\begin{align} 
\A(a^{(p)})(\psi(t)) &\;=\; \sum_{k = 0}^\infty \; \A\pb{G^{(k, 0)}_t (a^{(p)})}(\psi)
\notag \\ \label{Hartree evolution of observable}
&\;=\;
\sum_{k = 0}^\infty \frac{1}{2^k} \int_{\Delta^k(t)} \md \ul{t} \; \hb{\A(W_{t_k}), \dots \hb{\A(W_{t_1}), \A(a^{(p)}_t)}\dots}(\psi)
\,,
\end{align}
uniformly in the ball $B_\nu \deq \{\psi \in \LLL^2(\R^3) \,:\, \norm{\psi}^2 \leq \nu\}$.
\end{lemma}

\begin{proof}
The well-posedness of \eqref{integral Hartree} is a well-known result; see for instance \cite{ChadamGlassey, Zagatti}. The remaining statements follow from a ``tree expansion'', which also yields an existence result.
We first use the Schwinger-Dyson expansion to construct an evolution on the space of observables. We then show that this evolution stems from a Hamiltonian flow that satisfies the Hartree equation \eqref{integral Hartree}.

First, we generalize our class of ``observables'' to functions that are not gauge invariant, i.e.\ that correspond to bounded operators $a^{(q,p)} \in \mathcal{B}(\mathcal{H}_+^{p}; \mathcal{H}_+^{q})$. We set $\A(a^{(q,p)})(\psi) \deq \scalar{\psi^{\otimes q}}{a^{(q,p)} \psi^{\otimes p}}$, and denote by $\widetilde{\mathfrak{A}}$ the linear hull of observables of the form $\A(a^{(q,p)})$ with $a^{(q,p)} \in \mathcal{B}(\mathcal{H}_+^{p}; \mathcal{H}_+^{q})$.

It is convenient to introduce the abbreviations
\begin{equation*}
G \;\deq\; \{\A(h), \,\cdot\,\}\,, \qquad D \;\deq\; \frac{1}{2}\{\A(W), \,\cdot\,\}\,.
\end{equation*}
Then $\ee^{Gt}$ is well-defined on $\widetilde{\mathfrak{A}}$ through $(\me^{G t} A)(\psi) = A(\me^{-\mi h} \psi)$, where $A \in \widetilde{\mathfrak{A}}$. Note also that
\begin{equation*}
D_s \;\deq\; \ee^{Gs} D \ee^{-Gs} \;=\; \frac{1}{2}\{\A(W_s), \,\cdot\,\}\,.
\end{equation*}
Let $A \in \widetilde{\mathfrak{A}}$. We use the Schwinger-Dyson series for $\ee^{(G + D)t}$ to define the flow $S(t)A$ through
\begin{align}
S(t) A &\;\deq\; \sum_{k = 0}^\infty \int_{\Delta^k(t)} \md \ul{t} \; D_{t_k} \cdots D_{t_1} \, \me^{G t} A
\notag \\ \label{classical Schwinger-Dyson}
&\;=\; \sum_{k = 0}^\infty \int_{\Delta^k(t)} \md \ul{t} \; \frac{1}{2^k} \, \hb{\A(W_{t_k}), \dots \hb{\A(W_{t_1}), A_t)}\dots}\,.
\end{align}
Our first task is to show convergence of \eqref{classical Schwinger-Dyson} for small times.

Let $A = \A(a^{(q,p)})$. As with \eqref{Poisson bracket computed} one finds, after short computation, that
\begin{equation} \label{Poisson bracket of non gauge-invariant observable}
\frac{1}{2} \{\A(W), \A(a^{(q,p)})\} \;=\; \A \pbb{\mi \sum_{i = 1}^q W_{i \, q+1} (a^{(q,p)} \otimes \umat) - \mi \sum_{i = 1}^p (a^{(q,p)} \otimes \umat) W_{i \, p+1}}\,.
\end{equation}
Thus we see that the nested Poisson brackets in \eqref{classical Schwinger-Dyson} yield a ``tree expansion'' which may be described as follows. Define $T^{(k)}_{t, t_1, \dots, t_k}(a^{(q,p)})$ recursively through
\begin{align*}
T^{(0)}_t(a^{(q,p)}) &\;\deq\; a^{(q,p)}_t\,,
\\
T^{(k)}_{t, t_1, \dots, t_k}(a^{(q,p)}) &\;\deq\; \mi P_+ \sum_{i = 1}^{q + k-1} W_{i \, q+k, t_k} \pB{T^{(k-1)}_{t, t_1, \dots, t_{k-1}}(a^{(q,p)}) \otimes \umat} P_+ 
\\
&\qquad {}-{} \mi P_+ \sum_{i = 1}^{p + k-1} \pB{T^{(k-1)}_{t, t_1, \dots, t_{k-1}}(a^{(q,p)}) \otimes \umat} W_{i \, p+k, t_k} P_+\,.
\end{align*}
Note that $T^{(k)}_{t, t_1, \dots, t_k}(a^{(q,p)})$ is an operator from $\mathcal{H}_+^{(p+k)}$ to $\mathcal{H}_+^{(q+k)}$. Moreover, \eqref{Poisson bracket of non gauge-invariant observable} implies that
\begin{equation}
\frac{1}{2^k} \, \hb{\A(W_{t_k}), \dots \hb{\A(W_{t_1}), \A(a^{(q,p)}_t)}\dots} \;=\; \A \pB{T^{(k)}_{t, t_1, \dots, t_k}(a^{(q,p)})}\,.
\end{equation}
Also, by definition, we see that for gauge-invariant observables $a^{(p)}$ we have
\begin{equation*}
T^{(k)}_{t, t_1, \dots, t_k}(a^{(p)}) \;=\; G^{(k,0)}_{t, t_1, \dots, t_k}(a^{(p)})\,.
\end{equation*}
We may use the methods of Section \ref{section: Coulomb} to obtain the desired estimate. One sees that $T^{(k)}_{t, t_1, \dots, t_k}(a^{(p)})$ is a sum of elementary terms, indexed by labelled ordered trees, whose root has degree at most $p+q$, and whose other vertices have at most 3 children. From \eqref{sum over Catalan numbers} we find that there are
\begin{equation*}
\frac{p+q}{p+q+3k} \binom{p+q + 3k}{k}
\end{equation*}
unlabelled trees of this kind. Proceeding exactly as in Section \ref{section: Coulomb} we find that
\begin{equation*}
\int_{\Delta^k(t)} \md \ul{t} \; \normB{T^{(k)}_{t, t_1, \dots, t_k}(a^{(q,p)}) \Phi^{(p+k)}} \;\leq\; \binom{p+q + 3k}{k} \, \pbb{\frac{\pi \kappa^2 t}{2}}^{k/2} \, \norm{a^{(q,p)}} \norm{\Phi^{(p+k)}}\,,
\end{equation*}
where $\Phi^{(p+k)} \in \mathcal{H}_+^{(p+k)}$.
Let $\psi \in \LLL^2(\R^3)$ with $\norm{\psi}^2 \leq \nu$. Then $\abs{\A(a^{(q,p)}) (\psi)} \leq \norm{a^{(q,p)}} \norm{\psi}^{p+q}$ implies
\begin{multline} \label{bound for classical Schwinger-Dyson coefficients}
\int_{\Delta^k(t)} \md \ul{t} \; \absbb{\frac{1}{2^k} \, \hb{\A(W_{t_k}), \dots \hb{\A(W_{t_1}), \A(a^{(q,p)}_t)}\dots}(\psi)} 
\\
\leq\; \binom{p+q + 3k}{k} \, \pbb{\frac{\pi \kappa^2 t}{2}}^{k/2} \, \norm{a^{(q,p)}} \, \nu^{k + (p+q)/2}\,.
\end{multline}
Convergence of the Schwinger-Dyson series \eqref{classical Schwinger-Dyson} for small times $t$ follows immediately.

Thus, for small times $t$, the flow $S(t)$ is well-defined on $\widetilde{\mathfrak{A}}$, and it is easy to check that it satisfies the equation
\begin{equation} \label{abstract Hartree}
S(t) A \;=\; \me^{Gt} A + \int_0^t \md s \; S(s)\, D \, \me^{G (t-s)} A\,,
\end{equation}
for all $A \in \widetilde{\mathfrak{A}}$. 

In order to establish a link with the Hartree equation \eqref{integral Hartree}, we consider $f \in \LLL^2(\R^3)$ and define the function $F_f \in \widetilde{\mathfrak{A}}$ through $F_f(\psi) \deq \scalar{f}{\psi}$. Clearly, the mapping $f \mapsto (S(t) F_f)(\psi)$ is antilinear and \eqref{bound for classical Schwinger-Dyson coefficients} implies that it is bounded. Thus there exists a unique vector $\psi(t)$ such that
\begin{equation*}
(S(t) F_f)(\psi) \;\eqd\; \scalar{f}{\psi(t)}\,.
\end{equation*}
We now proceed to show that $(S(t) A)(\psi) = A(\psi(t))$ for all $A \in \widetilde{\mathfrak{A}}$. By definition, this is true for $A = F_f$. As a first step, we show that 
\begin{equation} \label{flow is a homomorphism}
S(t) (A B) \;=\; (S(t) A) (S(t) B)\,,
\end{equation}
where $A,B \in \widehat{\mathfrak{A}}$. Write
\begin{align*}
S(t) (AB) &\;=\; \sum_{k = 0}^\infty \int_{\Delta^k(t)} \md \ul{t} \; D_{t_k} \cdots D_{t_1} \, \me^{G t} (AB)
\\
&\;=\; \sum_{k = 0}^\infty \int_{\Delta^k(t)} \md \ul{t} \; D_{t_k} \cdots D_{t_1} \, (A_t B_t)\,,
\end{align*}
where we used $\me^{G t}(AB) = (\me^{Gt} A) (\me^{Gt} B)$. We now claim that
\begin{equation} \label{splitting of product}
\int_{\Delta^k(t)} \md \ul{t} \; D_{t_k} \cdots D_{t_1} (A_t B_t) \;=\; \sum_{l+m = k} \int_{\Delta^l(t)} \md \ul{t} \int_{\Delta^m(t)} \md \ul{s} \; \pb{D_{t_l} \cdots D_{t_1} A_t} \, \pb{D_{s_m} \cdots D_{s_1} B_t}\,,
\end{equation}
where the sum ranges over $l,m \geq 0$. This follows easily by induction on $k$ and using $D_s(A B) = A (D_s B) + (D_s A) B$. Then \eqref{flow is a homomorphism} follows immediately.

Next, we note that \eqref{flow is a homomorphism} implies that $(S(t) A)(\psi) = A(\psi(t))$, whenever $A$ is of the form $A = \A(a^{(q,p)})$, where
\begin{equation} \label{sum of tensor products}
a^{(q,p)} \;=\; \sum_j P_+ \ketb{f_1^j \otimes \cdots \otimes f_q^j} \brab{g_1^j \otimes \cdots \otimes g_p^j} P_+\,,
\end{equation}
where the sum is finite, and $f_i^j, g_i^j \in \LLL^2(\R^3)$. Now each $a^{(q,p)} \in \mathcal{B}(\mathcal{H}_+^{(p)}; \mathcal{H}_+^{(q)})$ can be written as the weak operator limit of a sequence $(a^{(q,p)}_n)_{n \in \N}$ of operators of type \eqref{sum of tensor products}. One sees immediately that
\begin{equation*}
\lim_n \A(a^{(q,p)}_n)(\psi(t)) \;=\; \A(a^{(q,p)})(\psi(t))\,.
\end{equation*}
On the other hand, uniform boundedness implies that $\sup_n \norm{a^{(q,p)}_n} < \infty$, so that
\begin{multline*}
\scalarB{\psi^{\otimes (q+k)}}{W_{i_1 j_1, t_{v_1}} \cdots W_{i_r j_r, t_{v_r}} \, \pb{a^{(q,p)}_n \otimes \umat^{(k)}} \, W_{i_{r+1} j_{r+1}, t_{v_{r+1}}} \cdots W_{i_k j_k, t_{v_k}}  \psi^{\otimes (p+k)}}
\\
\leq\;
\normb{a^{(q,p)}_n} \normB{W_{i_r j_r, t_{v_r}} \cdots W_{i_1 j_1, t_{v_1}} \psi^{\otimes (q+k)}}
\normB{W_{i_{r+1} j_{r+1}, t_{v_{r+1}}} \cdots W_{i_k j_k, t_{v_k}}  \psi^{\otimes (p+k)}}
\end{multline*}
justifies the use of dominated convergence in
\begin{equation*}
\lim_n (S(t) \A(a^{(q,p)}_n))(\psi) = (S(t) \A(a^{(q,p)}))(\psi)\,.
\end{equation*}
We have thus shown that
\begin{equation} \label{observable flow is Hamiltonian}
(S(t) A)(\psi) \;=\; A(\psi(t))\,, \qquad \forall A \in \widetilde{\mathfrak{A}}\,.
\end{equation}

Let us now return to \eqref{abstract Hartree}. Setting $A = F_f$, we find that \eqref{abstract Hartree} implies
\begin{align*}
\scalar{f}{\psi(t)} &\;=\; \scalar{f}{\me^{-ih} \psi} + \int_0^t \md s \; \frac{1}{2} \pB{S(s) \{\A(W), (F_f)_{t - s}\}}(\psi)
\\
&\;=\; \scalar{f}{\me^{-ih} \psi} + \int_0^t \md s \; \pb{\{\A(W), (F_f)_{t - s}\}}(\psi(s))\,,
\end{align*}
where we used \eqref{observable flow is Hamiltonian}. Using \eqref{Poisson bracket of non gauge-invariant observable} we thus find
\begin{equation} \label{Hartree f-component}
\scalar{f}{\psi(t)} \;=\; \scalar{f}{\me^{-ih} \psi} - \mi \int_0^t \md s \; \scalarb{(\me^{\mi h (t - s)} f) \otimes \psi(s)}{W \psi(s) \otimes \psi(s)}\,,
\end{equation}
which is exactly the Hartree equation \eqref{integral Hartree} projected onto $f$. We have thus shown that $\psi(t)$ as defined above solves the Hartree equation.

To show norm-conservation we abbreviate $F(s) \deq (w * \abs{\psi(s)}^2) \psi(s)$ and write, using \eqref{integral Hartree},
\begin{multline*}
\norm{\psi(t)}^2 - \norm{\psi}^2 \;=\; \mi \int_0^t \md s \; \qb{\scalarb{F(s)}{\me^{-\mi sh}\psi} - \scalarb{\me^{-\mi sh}\psi}{F(s)}} 
\\
{}+{} \int_0^t \md s \int_0^t \md r \; \scalarb{\me^{\mi sh}F(s)}{\me^{\mi rh}F(r)}\,.
\end{multline*}
The last term is equal to
\begin{equation*}
\int_0^t \md s \int_0^s \md r \; \qb{\scalarb{\me^{\mi sh}F(s)}{\me^{\mi rh}F(r)} + \scalarb{\me^{\mi rh}F(r)}{\me^{\mi sh}F(s)}}\,.
\end{equation*}
Therefore \eqref{integral Hartree} implies that
\begin{equation*}
\norm{\psi(t)}^2 - \norm{\psi}^2 \;=\; \mi \int_0^t \md s \; \scalarb{F(s)}{\psi(s)} - \mi \int_0^t \md s \; \scalarb{\psi(s)}{F(s)} \;=\; 0\,,
\end{equation*}
since $ \scalarb{F(s)}{\psi(s)} \in \R$, as can be seen by explicit calculation.
Thus we can iterate the above existence result for short times to obtain a global solution.

Furthermore, \eqref{Hartree f-component} implies that $\psi(t)$ is weakly continuous in $t$. Since the norm of $\psi(t)$ is conserved, $\psi(t)$ is strongly continuous in $t$. Similarly, the Schwinger-Dyson expansion \eqref{classical Schwinger-Dyson} implies that the map $\psi \mapsto \psi(t)$ is weakly continuous for small times, uniformly in $\norm{\psi}$ in compacts. Therefore, the map $\psi \mapsto \psi(t)$ is weakly continuous for all times $t$, and norm-conservation implies that it is strongly continuous.
\end{proof}

\subsection{Wick quantization}
In order to state our main result in a general setting, we shortly discuss how the many-body quantum mechanics of bosons can be viewed as a \emph{deformation quantization} of the (classical) Hartree theory. The deformation parameter (the analogue of $\hbar$ in the usual quantization of classical theories) is $1/N$. We define \emph{quantization} as the linear map $\widehat{(\cdot)}_N \,:\, \mathfrak{A} \to \widehat{\mathfrak{A}}$ defined by the formal replacement $\psi(x) \mapsto \widehat{\psi}_N(x)$ and $\bar{\psi}(x) \mapsto \widehat{\psi}^*_N(x)$ followed by Wick ordering. In other words,
\begin{equation*}
\widehat{(\cdot)}_N \,:\, \A(a^{(p)}) \;\mapsto\; \widehat{\A}_N(a^{(p)})\,.
\end{equation*}
Extending the definition of $\widehat{(\cdot)}_N$ to unbounded operators in the obvious way, we see that $\widehat{H}_N$ is the quantization of $H$.

Note that \eqref{product of two second quantized operators} and \eqref{Poisson bracket computed} imply, for $A,B \in \mathfrak{A}$,
\begin{equation*}
\qb{\widehat{A}_N, \widehat{B}_N} \;=\; \frac{N^{-1}}{\ii} \widehat{\h{A, B}}_N + O\pbb{\frac{1}{N^2}}\,,
\end{equation*}
so that $1/N$ is indeed the deformation parameter of $\widehat{(\cdot)}_N$.

\subsection{The mean-field limit: a Egorov-type result}
Let $\phi^t$ denote the Hamiltonian flow of the Hartree equation on $\LLL^2(\R^3)$.
Introduce the short-hand notation
\begin{align*}
\alpha^t A &\;\deq\; A \circ \phi^t\,, &A \in \mathfrak{A}\,,
\\
\widehat{\alpha}^t \mathbf{A} &\;\deq\; \me^{\mi t N \widehat{H}_N} \, \mathbf{A} \, \me^{-\mi t N \widehat{H}_N}\,, &\mathbf{A} \in \widehat{\mathfrak{A}}\,.
\end{align*}
We may now state and prove our main result, which essentially says that, in the mean-field limit $n = \nu N \to \infty$, time evolution and quantization commute.
\begin{theorem}
Let $A \in \mathfrak{A}$, $\nu > 0$, and $\epsilon > 0$. Then there exists a function $A(t) \in \mathfrak{A}$ such that
\begin{equation*}
\sup_{t \in \R} \normb{\alpha^t A - A(t)}_{\LLL^\infty(B_\nu)} \;\leq\; \epsilon\,,
\end{equation*}
as well as
\begin{equation*}
\normb{\pb{\widehat{\alpha}^t \widehat{A}_N - \widehat{A(t)}_N} \bigr|_{\mathcal{H}^{(\nu N)}_+}} \;\leq\; \epsilon + \frac{C(\epsilon, \nu, t, A)}{N}\,. 
\end{equation*}
\end{theorem}
\begin{remark*}
The ``intermediate function'' $A(t)$ is required, since the full time evolution $\alpha^t$ does not leave $\mathfrak{A}$ invariant.
\end{remark*}

\begin{proof}
Most of the work has already been done in the previous sections. Without loss of generality take $A = \A(a^{(p)})$ for some $p \in \N$ and $a^{(p)} \in \mathcal{B}(\mathcal{H}^{(p)}_\pm)$. Assume that $t < \rho(\kappa, \nu)$. Taking $L = 1$ in \eqref{1/N expansion} we get
\begin{equation} \label{one-loop expansion}
\widehat{\alpha}^t \, \widehat{\A}_N(a^{(p)})  \Bigr|_{\mathcal{H}^{(\nu N)}_+} \;=\; \sum_{k = 0}^\infty \, \widehat{\A}_N\pb{G^{(k, 0)}_t(a^{(p)})} \Bigr|_{\mathcal{H}^{(\nu N)}_+} + O \pbb{\frac{1}{N}}\,.
\end{equation}
Comparing this with \eqref{Hartree evolution of observable} immediately yields
\begin{equation*}
\widehat{\alpha}^t \widehat{\A}_N(a^{(p)}) \;=\; \qb{\alpha^t \A(a^{(p)})}\quant + O\pbb{\frac{1}{N}}
\end{equation*}
on $\mathcal{H}^{(\nu N)}_+$, where $\qb{\alpha^t \A(a^{(p)})}\quant$ is defined through its norm-convergent power series. This is the statement of the theorem for short times.

The extension to all times follows from an iteration argument. We postpone the details to the proof of Theorem \ref{theorem: bosons mean-field for pure states} below. In its notation $A(t)$ is given by
\begin{equation*}
A(t) \;=\; \sum_{k_1 = 0}^{K_1 - 1} \cdots \sum_{k_m = 0}^{K_m - 1} \A\pb{G^{(k_m, 0)}_{\tau} G^{(k_{m-1}, 0)}_\tau \cdots G^{(k_1, 0)}_\tau a^{(p)}}\,.
\qedhere
\end{equation*}
\end{proof}

The result may also be expressed in terms of coherent states. 
\begin{theorem} \label{theorem: bosons mean-field for pure states}
Let $a^{(p)} \in \mathcal{B}(\mathcal{H}^{(p)}_+)$, $\psi \in \LLL^2(\R^3)$ with $\norm{\psi} = 1$, and $T > 0$. Then there exist constants $C, \beta > 0$, depending only on $p$, $T$ and $\kappa$, such that
\begin{equation} \label{boson mean-field limit for coherent states}
\absB{\scalarB{\psi^{\otimes N}}{\me^{\mi t H_N} \, \widehat{\A}_N(a^{(p)}) \, \me^{-\mi t H_N} \, \psi^{\otimes N}} - \scalarb{\psi(t)^{\otimes p}}{a^{(p)} \psi(t)^{\otimes p}}} \;\leq\; \frac{C}{N^\beta } \, \norm{a^{(p)}}\,, \qquad t \in [0,T]\,.
\end{equation}
Here $\psi(t)$ is the solution to the Hartree equation \eqref{integral Hartree} with initial data $\psi$.
\end{theorem}
\begin{proof}
Introduce a cutoff $K \in \N$ and write (in self-explanatory notation)
\begin{align} \label{splitting of the quantum evolution}
\widehat{\alpha}^\tau \widehat{\A}_N(a^{(p)}) &\;=\; \sum_{k = 0}^{K-1} \widehat{\A}_N\pb{G_\tau^{(k,0)}(a^{(p)})} + \widehat{\alpha}^\tau_{\geq K} \widehat{\A}_N(a^{(p)}) + \frac{1}{N} R_{N,\tau}(a^{(p)})\,,
\\
\alpha^\tau \A(a^{(p)}) &\;=\; \sum_{k = 0}^{K-1} \A \pb{G_\tau^{(k,0)}(a^{(p)})} + \alpha^\tau_{\geq K} \A(a^{(p)})\,.
\end{align}
To avoid cluttering the notation, from now on we drop the parentheses of the linear map $G^{(k,0)}_\tau$. We iterate \eqref{splitting of the quantum evolution} $m$ times by applying it to its first term and get
\begin{multline} \label{full iterated expansion}
(\widehat{\alpha}^\tau)^m \widehat{\A}_N(a^{(p)}) \;=\;
\sum_{k_1 = 0}^{K_1 - 1} \cdots \sum_{k_m = 0}^{K_m - 1} \widehat{\A}_N\pB{G^{(k_m, 0)}_\tau G^{(k_{m-1}, 0)}_\tau \cdots G^{(k_1, 0)}_\tau a^{(p)}}
\\
{}+{} (\widehat{\alpha}^\tau)^{m-1} \widehat{\alpha}^\tau_{\geq K_1} \widehat{\A}_N(a^{(p)}) + \sum_{j = 1}^{m-1} \sum_{k_1 = 0}^{K_1 - 1} \cdots \sum_{k_j = 0}^{K_j - 1} (\widehat{\alpha}^\tau)^{m -1 - j} \widehat{\alpha}^\tau_{\geq K_{j+1}} \widehat{\A}_N\pB{G^{(k_j, 0)}_\tau G^{(k_{j-1}, 0)}_\tau \cdots G^{(k_1, 0)}_\tau a^{(p)}}
\\
{}+{} \frac{1}{N} (\widehat{\alpha}^\tau)^{m-1} R_{N,\tau}(a^{(p)}) + \frac{1}{N} \sum_{j = 1}^{m-1}\sum_{k_1 = 0}^{K_1 - 1} \cdots \sum_{k_j = 0}^{K_j - 1} (\widehat{\alpha}^\tau)^{m-1-j} R_{N,\tau}
\pb{G^{(k_j, 0)}_\tau \cdots G^{(k_1, 0)}_\tau a^{(p)}}\,.
\end{multline}
A similar expression without the third line holds for $(\alpha^\tau)^m \A(a^{(p)})$.

In order to control this somewhat unpleasant expression, we abbreviate 
\begin{equation*}
x \deq \sqrt{\frac{\tau}{\rho(\kappa,1)}}\,.
\end{equation*}
Assume that $x < 1$. Then \eqref{fundamental estimate} and \eqref{1/N expansion} imply the estimates, valid on $\mathcal{H}^{(N)}_+$,
\begin{align*}
\normb{G^{(k,0)}_\tau \, a^{(p)}} &\;\leq\; 4^p \norm{a^{(p)}} \, x^k\,,
\\
\normb{\widehat{\alpha}^\tau_{\geq K} \widehat{\A}_N(a^{(p)})} &\;\leq\; 4^p \norm{a^{(p)}} \, \frac{x^{K}}{1-x}\,,
\\
\normb{R_{N,\tau}(a^{(p)})} &\;\leq\; (4 \me)^p \norm{a^{(p)}} \, \frac{x}{(1-x)^3}\,.
\end{align*}
Furthermore, \eqref{Hartree evolution of observable} implies that
\begin{equation*}
\normb{\alpha^\tau_{\geq K} \A(a^{(p)})}_{\LLL^\infty(B_1)} \;\leq\; 4^p \norm{a^{(p)}} \, \frac{x^{K}}{1-x}\,.
\end{equation*}
We also need
\begin{align}
\absb{\scalarb{\psi^{\otimes N}}{\widehat{\A}_N(a^{(p)}) \psi^{\otimes N}} - \A(a^{(p)})(\psi)} &\;=\;\absbb{\frac{N \cdots (N-p+1)}{N^p} - 1} \absb{\A(a^{(p)})(\psi)}
\notag \\
&\;\leq\; \sum_{j = 1}^{p-1} \absbb{\frac{N \cdots (N-j)}{N^{j+1}} - \frac{N \cdots (N-j+1)}{N^j}} \norm{a^{(p)}}
\notag \\ \label{estimate of quantization error}
&\;\leq\; \frac{p^2}{N} \norm{a^{(p)}}\,.
\end{align}
Armed with these estimates we may now complete the proof of Theorem \ref{theorem: bosons mean-field for pure states}.
Suppose that $1/2 \leq x  < 1$. Then
\begin{multline*}
\sum_{k_1 = 0}^{K_1 - 1} \cdots \sum_{k_m = 0}^{K_m - 1} \absB{\scalarB{\psi^{\otimes N}}{\widehat{\A}_N\pB{G^{(k_m, 0)}_\tau G^{(k_{m-1}, 0)}_\tau \cdots G^{(k_1, 0)}_\tau a^{(p)}} \, \psi^{\otimes N}}
\\
{}-{} \A\pB{G^{(k_m, 0)}_\tau G^{(k_{m-1}, 0)}_\tau \cdots G^{(k_1, 0)}_\tau a^{(p)}} (\psi)
}
\\
\leq\; \frac{1}{N} (p + K_1 + \cdots + K_m)^2 \, 4^{m (p+ K_1 + \cdots + K_m)} \, \norm{a^{(p)}}\,.
\end{multline*}
Similarly, the second line of \eqref{full iterated expansion} on $\mathcal{H}^{(N)}_+$ and its classical equivalent on $B_1$ are bounded by
\begin{equation*}
\sum_{j = 1}^m x^{K_j} \, 4^{j(p+K_1 + \cdots + K_{j - 1})} \, \norm{a^{(p)}}\,.
\end{equation*}
Finally, the last line of \eqref{full iterated expansion} on $\mathcal{H}^{(N)}_+$ is bounded by
\begin{equation*}
\frac{1}{N} \sum_{j = 1}^m 4^{(j+1)(p+K_1 + \cdots + K_{j - 1})} \, \norm{a^{(p)}}\,.
\end{equation*}

Now pick $m$ large enough that $T \leq m \tau$. Then it is easy to check that there exist $a_1, \dots, a_m$ such that setting
\begin{equation*}
K_j \;=\; a_j \, \log N \,, \qquad j = 1, \dots, m
\end{equation*}
implies that the three above expressions are all bounded by $C N^{-\beta} \norm{a^{(p)}}$, for some $\beta > 0$. This remains of course true for all $m' \leq m$. Since any time $t \leq T$ can be reached by at most $m$ iterations with $1/2 \leq x < 1$, the claim follows.
\end{proof}

We conclude with a short discussion on density matrices. First we recall some standard results; see for instance \cite{ReedSimonI}. Let $\Gamma \in \mathcal{L}^1$, where $\mathcal{L}^1$ is the space of trace class operators on some Hilbert space. Equipped with the norm $\norm{\Gamma}_1 \deq \tr \abs{\Gamma}$, $\mathcal{L}^1$ is a Banach space. Its dual is equal to $\mathcal{B}$, the space of bounded operators, and the dual pairing is given by
\begin{equation*}
\scalar{A}{\Gamma} \;=\; \tr (A \Gamma)\,, \qquad A \in \mathcal{B}\,, \Gamma \in \mathcal{L}^1\,.
\end{equation*}
Therefore,
\begin{equation} \label{dual representation of trace norm}
\norm{\Gamma}_1 \;=\; \sup_{A \in \mathcal{B},\, \norm{A} \leq 1} \abs{\tr (A \Gamma)}\,.
\end{equation}

Consider an $N$-particle density matrix $0 \leq \Gamma_N \in \mathcal{L}^1(\mathcal{H}^{(N)}_+)$ that satisfies $\tr \Gamma_N = 1$ and is symmetric in the sense that $\Gamma_N P_+ = \Gamma_N$. Define the $p$-particle marginals
\begin{equation*}
\Gamma_N^{(p)} \;\deq\; 
\tr_{p+1, \dots, N} \Gamma_N\,,
\end{equation*}
where $\tr_{p+1, \dots, N}$ denotes the partial trace over the coordinates $p+1, \dots, N$.
Define furthermore
\begin{equation*}
\Gamma_N(t) \;=\; \me^{-\mi t H_N} \Gamma_N \me^{\mi t H_N}\,,
\end{equation*}
as well as the $p$-particle marginals $\Gamma_N^{(p)}(t)$ of $\Gamma_N(t)$.

Noting that
\begin{equation*}
\tr \pB{\widehat{\A}_N(a^{(p)}) \, \Gamma_N(t)} \;=\; \frac{p!}{N^p} \binom{N}{p} \, \tr \pb{a^{(p)} \Gamma_N^{(p)}(t)}
\;=\; \tr \pb{a^{(p)} \Gamma_N^{(p)}(t)} + O \pbb{\frac{1}{N}}
\end{equation*}
we see that \eqref{dual representation of trace norm} and Theorem \ref{theorem: bosons mean-field for pure states} imply the following result.
\begin{corollary} \label{cor: density matrices}
Let $\psi \in \mathcal{H}$ with $\norm{\psi} = 1$, and let $\psi(t)$ be the solution of \eqref{integral Hartree} with initial data $\psi$. Set $\Gamma_N \deq (\ket{\psi} \bra{\psi})^{\otimes N}$. Then, for any $p \in \N$ and $T > 0$ there exist constants $C, \beta > 0$, depending only on $p$, $T$ and $\kappa$, such that
\begin{equation*}
\normB{\Gamma_N^{(p)}(t) - \pb{\ket{\psi(t)} \bra{\psi(t)}}^{\otimes p}}_1 \;\leq\; \frac{C}{N^\beta}\,, \qquad t \in [0,T]\,.
\end{equation*}
\end{corollary}

\begin{remark*}
Actually it is enough for $\Gamma_N$ to factorize asymptotically. If
$(\Gamma_N)_{N \in \N}$ is a sequence of symmetric density matrices satisfying
\begin{equation*}
\lim_{N \to \infty} \normb{\Gamma_N^{(1)} - \ket{\psi} \bra{\psi}}_1 \;=\; 0\,,
\end{equation*}
then one finds
\begin{equation*}
\lim_{N\to 0} \normB{\Gamma_N^{(1)}(t) - \ket{\psi(t)} \bra{\psi(t)}}_1 \;=\; 0\,, \qquad t \in \R\,.
\end{equation*}
This is a straightforward corollary of the proof of Theorem \ref{theorem: bosons mean-field for pure states}. By an argument of Lieb and Seiringer (see the remark after Theorem 1 in \cite{LiebSeiringer}), this implies that
\begin{equation*}
\lim_{N\to 0} \normB{\Gamma_N^{(p)}(t) - \pb{\ket{\psi(t)} \bra{\psi(t)}}^{\otimes p}}_1 \;=\; 0\,, \qquad t \in \R\,.
\end{equation*}
for all $p$.
\end{remark*}

\section{Some Generalizations} \label{section: generalizations}
In this section we generalize our results to a larger class of interaction potentials, and allow an external potential. For this we need Strichartz estimates for Lorentz spaces. We start with a short summary of the relevant results (see \cite{BerghLofstrom1976, KeelTao1998}).

For $1 \leq q \leq \infty$ and $0<\theta < 1$ we define the real interpolation functor $(\cdot, \cdot)_{\theta, q}$ as follows. Let $A_0$ and $A_1$ be two Banach spaces contained in some larger Banach space $A$. 
Define the real interpolation norm
\begin{equation*}
\norm{a}_{(A_0, A_1)_{\theta, q}} \;\deq\; 
\begin{cases}
\qB{\int_0^\infty \pb{t^{-\theta} K(t,a)}^q \, \md t/t}^{1/q}\,, & q \;<\; \infty\,,
\\
\sup_{t\geq 0} t^{-\theta} K(t, a)\,, & q \;=\; \infty\,.
\end{cases}
\end{equation*}
where
\begin{equation*}
K(t,a) \;\deq\; \inf_{a = a_0 + a_1} \pb{\norm{a_0}_{A_0} + t \norm{a_1}_{A_1}}\,.
\end{equation*}
Define $(A_0, A_1)_{\theta, q}$ as the space of $a \in A$ such that $\norm{a}_{(A_0, A_1)_{\theta, q}} < \infty$. Then $(A_0, A_1)_{\theta, q}$ is a Banach space.
The \emph{Lorentz space} $\LLL^{p,q}(\R^3, \md x) \equiv \LLL^{p,q}$ is defined by interpolation as
\begin{equation*}
\LLL^{p,q} \;\deq\; (\LLL^{p_0}, \LLL^{p_1})_{\theta, q}\,,
\end{equation*}
where $1 \leq p_0, p_1 \leq \infty$, $p_0 \neq p_1$, and
\begin{equation*}
\frac{1}{p} \;=\; \frac{1 -\theta}{p_0} + \frac{\theta}{p_1}\,.
\end{equation*}

Lorentz spaces have the following properties that are of interest to us. First, $\LLL^{p,p} = \LLL^p$. Second, $\LLL^{p,\infty} = \LLL^{p}_w$, where $\LLL^p_w$ is the weak $\LLL^p$ space (see e.g.\ \cite{ReedSimonII, BerghLofstrom1976}). In particular,  the Coulomb potential in 3 dimensions satisfies
\begin{equation*}
\frac{1}{\abs{x}} \;\in\; \LLL^{3, \infty}\,.
\end{equation*}
Finally, Lorentz spaces satisfy a general H\"older inequality (see \cite{ONeil1963}): Let $1 < p,p_1,p_2 < \infty$
 and $1 \leq q, q_1, q_2 \leq \infty$ satisfy
\begin{equation*}
\frac{1}{p_1} + \frac{1}{p_2} = \frac{1}{p} \,,\qquad \frac{1}{q_1} + \frac{1}{q_2} = \frac{1}{q}\,.
\end{equation*}
Then we have
\begin{equation} \label{general Holder}
\norm{fg}_{\LLL^{p,q}} \;\lesssim\; \norm{f}_{\LLL^{p_1, q_1}} \norm{g}_{\LLL^{p_2, q_2}}\,.
\end{equation}

We need an endpoint homogeneous Strichartz estimate proved in \cite{KeelTao1998}. For a map $f : \R \to \LLL^{p,q}$ we define the space-time norm
\begin{equation*}
\norm{f}_{\LLL^r_t \LLL^{p,q}_x} \;\deq\; \qbb{\int \md t \; \norm{f(t)}_{\LLL^{p,q}}^r}^{1/r}\,.
\end{equation*}
Then Theorem 10.1 of \cite{KeelTao1998} implies that
\begin{equation} \label{Strichartz estimate}
\normb{\me^{\mi t \Delta} f}_{\LLL^r_t \LLL^{p,2}_x} \;\lesssim\; \norm{f}_{\LLL^2}\,,
\end{equation}
whenever $2 \leq r < \infty$ and
\begin{equation*}
\frac{2}{r} + \frac{3}{p} \;=\; \frac{3}{2}\,.
\end{equation*}
We are now set for proving a generalization of \eqref{kato smoothing estimate}.
\begin{lemma} \label{lemma: generalized Kato smoothing}
Let $w \in \LLL^3_w + \LLL^\infty$. Then there is a constant $C = C(w) > 0$, such that
\begin{equation*}
\int_0^1 \norm{w \, \me^{\mi t \Delta} \, \psi}^2 \, \md t\;\leq\; C \norm{\psi}^2\,.
\end{equation*}
\end{lemma}
\begin{proof}
Let $w = w_1 + w_2$ with $w_1 \in \LLL^\infty$ and $w_2 \in \LLL^3_w$. Then
\begin{equation*}
\normb{w \, \me^{\mi t \Delta} \, \psi}_{\LLL^2_t \LLL^2_x} \;\leq\;
\normb{w_1 \, \me^{\mi t \Delta} \, \psi}_{\LLL^2_t \LLL^2_x} + \normb{w_2 \, \me^{\mi t \Delta} \, \psi}_{\LLL^2_t \LLL^2_x}\,.
\end{equation*}
The first term is bounded by $\norm{w_1}_{\LLL^\infty} \norm{\psi}_{\LLL^2}$. To bound the second we use \eqref{general Holder} and \eqref{Strichartz estimate} with $r = 2$ and $p = 6$ to get
\begin{equation*}
\normb{w_2 \, \me^{\mi t \Delta} \, \psi}_{\LLL^2_t \LLL^2_x} \;\lesssim\; \norm{w_2}_{\LLL^{3,\infty}} \normb{\me^{\mi t \Delta} \, \psi}_{\LLL^2_t \LLL^{6,2}_x} \;\lesssim\; \norm{w_2}_{\LLL^{3,\infty}} \norm{\psi}_{\LLL^2}\,.
\end{equation*}
Therefore,
\begin{equation*}
\normb{w \, \me^{\mi t \Delta} \, \psi}_{\LLL^2_t \LLL^2_x} \;\leq\; \sqrt{C(w)} \, \norm{\psi}_{\LLL^2}\,.
\qedhere
\end{equation*}
\end{proof}

Now let us assume that $v,w \in \LLL^\infty + \LLL^3_w$. 
Set $H_0 |_{\mathcal{H}^{(n)}_\pm} \deq \sum_{i = 1}^n -\Delta_i$. Then the required generalization of Lemma \ref{lemma: Kato smoothing for centre of mass} is
\begin{lemma} \label{lemma: generalized Kato smoothing for centre of mass}
There exists a constant $C \equiv C(w, v)$ such that
\begin{align*}
\int_0^1 \normb{W_{ij} \, \me^{-\mi t H_0} \Phi^{(n)}}^2 \md t &\;\leq\; C \norm{\Phi^{(n)}}^2\,,
\\
\int_0^1 \normb{V_i \, \me^{-\mi t H_0} \Phi^{(n)}}^2 \md t &\;\leq\; C \norm{\Phi^{(n)}}^2\,,
\end{align*}
where $\Phi^{(n)} \in \mathcal{H}^{(n)}_\pm$\,.
\end{lemma}
\begin{proof}
The claim for $V$ follows immediately from Lemma \ref{lemma: generalized Kato smoothing}. The estimate for $W$ follows similarly by using centre of mass coordinates.
\end{proof}

Finally, we briefly discuss the changes to the combinatorics arising from an external potential. We classify the elementary terms according to the numbers $(k,l,m)$, where $k$ is the order of the multiple commutator, $l$ is the number of loops, and $m$ is the number of $V$-operators. Thus, instead of \eqref{recursive definition of graphs}, we have the recursive definition
\begin{align*}
G^{(k,l,m)}_{t,t_1,\dots,t_k}(a^{(p)})
&\;=\; \mi (p+k-l-m-1) \qB{W_{t_k} , G^{(k-1,l,m)}_{t, t_1, \dots, t_{k - 1}}(a^{(p)})}_1
\\
&\qquad{}+{} \mi \binom{p+k-l-m}{2} \qB{W_{t_k}, G^{(k-1, l-1,m)}_{t, t_1, \dots, t_{k - 1}}(a^{(p)})}_2
\\
&\qquad{}+{}\mi (p+k-l-m) \qB{V_{t_k} , G^{(k-1,l,m-1)}_{t, t_1, \dots, t_{k - 1}}(a^{(p)})}_1
\\
&\;=\; \mi P_\pm  \sum_{i = 1}^{p+k-l-m-1} \qb{W_{i \, p+k-l-m, t_k}, G^{(k-1,l,m)}_{t, t_1, \dots, t_{k - 1}}(a^{(p)}) \otimes \umat} P_\pm 
\\
&\qquad{}+{} \mi P_\pm  \sum_{1 \leq i < j \leq p+k-l-m} \qb{W_{ij, t_k}, G^{(k-1, l-1,m)}_{t, t_1, \dots, t_{k - 1}}(a^{(p)})} P_\pm  
\\
&\qquad{}+{} \mi P_\pm  \sum_{i = 1}^{p+k-l-m} \qb{V_{i, t_k}, G^{(k-1, l,m-1)}_{t, t_1, \dots, t_{k - 1}}(a^{(p)})}P_\pm \,,
\end{align*}
as well as $G^{(0,0,0)}_t(a^{(p)}) \deq a^{(p)}_t$. We also set $G^{(k,l,m)}_{t,t_1,\dots,t_k}(a^{(p)}) = 0$ unless $0 \leq l \leq k-m$. 
It is again an easy exercise to show by induction on $k$ that 
\begin{equation*}
\frac{(\mi N)^k}{2^k} \qB{\widehat{\A}_N(W_{t_k}), \dots \qB{\widehat{\A}_N(W_{t_1}), \widehat{\A}_N(a^{(p)}_t)}\dots} \;=\; \sum_{l = 0}^k \sum_{m = 0}^{k - l} \frac{1}{N^l} \, \widehat{\A}_N \pb{ G^{(k,l,m)}_{t,t_1,\dots,t_k}(a^{(p)})}\,.
\end{equation*}
Note that $G^{(k,l,m)}_{t,t_1,\dots,t_k}(a^{(p)})$ is a $p + k - l - m$ particle operator.

The graphs of Section \ref{section: Coulomb} have to be modified: Each vertex corresponding to a $V$-operator has one edge for each direction $d = a,c$ (see Figure \ref{figure: general graph with external potential}).
\begin{figure}[ht!]
\psfrag{0}[][]{$0$}
\psfrag{1}[][]{$1$}
\psfrag{2}[][]{$2$}
\psfrag{3}[][]{$3$}
\psfrag{4}[][]{$4$}
\psfrag{5}[][]{$5$}
\psfrag{6}[][]{$6$}
\psfrag{7}[][]{$7$}
\psfrag{a1}[][]{$\scriptstyle{a,1}$}
\psfrag{a2}[][]{$\scriptstyle{a,2}$}
\psfrag{a3}[][]{$\scriptstyle{a,3}$}
\psfrag{a4}[][]{$\scriptstyle{a,4}$}
\psfrag{c1}[][]{$\scriptstyle{c,1}$}
\psfrag{c2}[][]{$\scriptstyle{c,2}$}
\psfrag{c3}[][]{$\scriptstyle{c,3}$}
\psfrag{c4}[][]{$\scriptstyle{c,4}$}
\vspace{0.5cm}
\begin{center}
\includegraphics[width = 12cm]{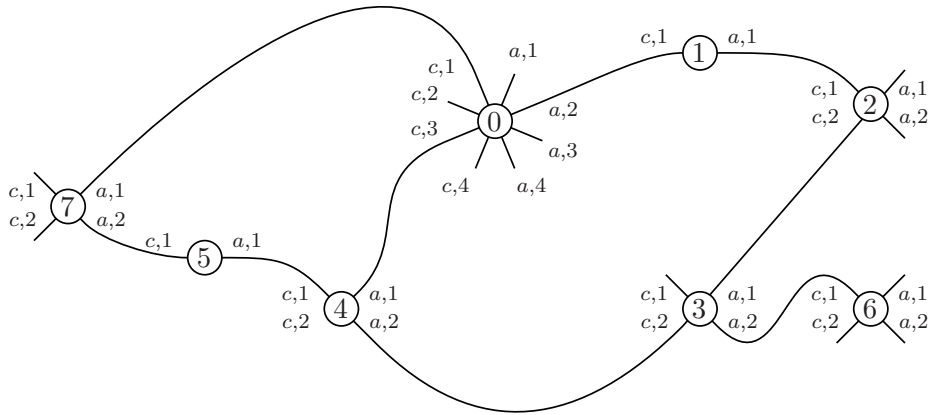}
\end{center}
\caption{\textit{An admissible graph of type $(p = 4, k = 7, l = 2, m = 2)$.} \label{figure: general graph with external potential}}
\end{figure}

Let us first consider tree graphs, $l = 0$. Take the set of trees without an external potential as in Section \ref{section: Coulomb}. By allowing each vertex $v = 1, \dots, k$ whose edges $(a, 2)$ and $(c, 2)$ are empty to stand for either an interaction potential $W$ or an external potential $V$, we count all trees with an external potential. Thus, for a given $m$, there are at most $\binom{k}{m} \abs{\mathscr{G}(p,k,0)}$ tree graphs contributing to $G^{(k,0,m)}_{t,t_1,\dots,t_k}(a^{(p)})$. If $l > 0$ we repeat the argument in the proof Lemma \ref{lemma: bound on l-loop graphs}, and find that the number of graph structures contributing to $G^{(k,l,m)}_{t,t_1,\dots,t_k}(a^{(p)})$ is bounded by
\begin{equation*}
2^k \, \binom{k}{m} \, \binom{k}{l} \, \binom{2p + 3k}{k} \, (p+k-l-m)^l\,.
\end{equation*}
Putting all this together, we find that
\begin{equation*}
\normb{G^{(k,l,m)}_t(a^{(p)})} \;\leq\;
\binom{k}{m} \binom{k}{l} \, \binom{2p+3k}{k} \,  (p + k - l - m)^l (C t)^{k/2}\, \norm{a^{(p)}}\,.
\end{equation*}
Using the condition $p+k-l-m \leq n$, it is then easy to see that all convergence estimates remain valid with the additional factor $2^k$.

In summary, all of the results of Sections \ref{section: Coulomb} and \ref{section: mean-field limit for bosons} hold if
\begin{equation*}
v,w \;\in\; \LLL^3_w + \LLL^\infty\,.
\end{equation*}

\appendix
\section{Second Quantization} \label{second quantization}
We briefly summarize the main ingredients of many-body quantum mechanics and second quantization. See for instance \cite{BratteliRobinsonII} for an extensive discussion.

Let $\mathcal{H} = \LLL^2(\R^d, \dd x)$ be the ``one-particle Hilbert space'', where $d \in \N$. Many-body quantum mechanics is formulated on subspaces of the $n$-particle spaces $\mathcal{H}^{\otimes n}$. Let $P^{(n)}_\pm \equiv P_{\pm}$ be the orthogonal projector onto the symmetric/antisymmetric subspace of $\mathcal{H}^{\otimes n}$, i.e.
\begin{equation*}
(P_{\pm} \Phi^{(n)})(x_1, \dots, x_n) \;\deq\; \frac{1}{n!} \sum_{\sigma \in S_n} (\pm 1)^{\abs{\sigma}}\Phi^{(n)}(x_{\sigma(1)}, \dots, x_{\sigma(n)})\,,
\end{equation*}
where $\abs{\sigma}$ denotes the number of transpositions in the permutation $\sigma$, and $\Phi^{(n)} \in \mathcal{H}^{\otimes n}$. We define the bosonic $n$-particle space as $\mathcal{H}^{(n)}_+ \deq P_+ \mathcal{H}^{\otimes n}$, and the fermionic $n$-particle space as $\mathcal{H}^{(n)}_- \deq P_- \mathcal{H}^{\otimes n}$. We adopt the usual convention that $\mathcal{H}^{\otimes 0} = \C$. 

We introduce the Fock space
\begin{equation*}
\mathcal{F}_\pm(\mathcal{H}) \;\equiv\; \mathcal{F}_\pm \;\deq\; \bigoplus_{n = 0}^\infty \mathcal{H}^{(n)}_\pm\,.
\end{equation*}
A state $\Phi \in \mathcal{F}_\pm$ is a sequence $\Phi = (\Phi^{(n)})_{n = 0}^\infty$, where $\Phi^{(n)} \in \mathcal{H}^{(n)}_\pm$. Equipped with the scalar product
\begin{equation*}
\scalar{\Phi}{\Psi} \;=\; \sum_{n = 0}^\infty \scalarb{\Phi^{(n)}}{\Psi^{(n)}}
\end{equation*}
$\mathcal{F}_\pm$ is a Hilbert space.
The vector $\Omega \deq (1, 0, 0, \dots)$ is called the vacuum. By a slight abuse of notation, we denote a vector of the form $\Phi = (0, \dots, 0, \Phi^{(n)}, 0, \dots) \in \mathcal{F}_\pm$ by its non-vanishing $n$-particle component $\Phi^{(n)}$.
Define also the subspace of vectors with a finite particle number
\begin{equation*}
\mathcal{F}^0_\pm \;\deq\; \h{\Phi \in \mathcal{F}_\pm \,:\, \Phi^{(n)} = 0 \text{ for all but finitely many } n}\,.
\end{equation*}

On $\mathcal{F}_\pm$ we have the usual creation and annihilation operators, $\widehat{\psi}^*$ and $\widehat{\psi}$, which map the one-particle space $\mathcal{H}$ into densely defined closable operators on $\mathcal{F}_\pm$. For $f \in \mathcal{H}$ and $\Phi \in \mathcal{F}_\pm$, they are defined by
\begin{align*}
\pb{\widehat{\psi}^*(f) \Phi}^{(n)}(x_1, \dots, x_n) &\;\deq\;
\frac{1}{\sqrt{n}}\sum_{i = 1}^{n} (\pm 1)^{i - 1} \, f(x_i)
\Phi^{(n-1)}(x_1, \dots, x_{i-1}, x_{i+1}, \dots, x_n)\,,
\\
\pb{\widehat{\psi}(f) \Phi}^{(n)}(x_1, \dots, x_n) &\;\deq\; \sqrt{n+1} \int \md y \; \bar{f}(y) \Phi^{(n+1)}(y, x_1, \dots, x_n)\,.
\end{align*}
It is not hard to see that $\widehat{\psi}(f)$ and $\widehat{\psi}^*(f)$ are adjoints of each other (see for instance \cite{BratteliRobinsonII} for details). Furthermore, they satisfy the canonical (anti)commutation relations
\begin{equation} \label{unrescaled commutation relations}
\qb{\widehat{\psi}(f), \widehat{\psi}^*(g)}_\mp \;=\; \scalar{f}{g}\, \umat\,,\qquad \qb{\widehat{\psi}^\#(f), \widehat{\psi}^\#(g)}_\mp \;=\; 0\,,
\end{equation}
where $[A, B]_\mp \deq AB \mp BA$, and $\widehat{\psi}^\# = \widehat{\psi}^*$ or $\widehat{\psi}$. In order to simplify notation, we usually identify $c \umat$ with $c$, where $c \in \C$.

For our purposes, it is more natural to work with the rescaled creation and annihilation operators
\begin{equation*}
\widehat{\psi}^\#_N \;\deq\; \frac{1}{\sqrt{N}} \, \widehat{\psi}^\#\,,
\end{equation*}
where $N > 0$. We also introduce the operator-valued distributions defined formally by
\begin{equation*}
\widehat{\psi}^\#_N(x) \;\deq\; \widehat{\psi}^\#_N(\delta_x)\,,
\end{equation*}
where $\delta_x$ is the delta function at $x$. The formal expression $\widehat{\psi}^\#_N(x)$ has a rigorous meaning as a densely defined sesquilinear form on $\mathcal{F}_\pm$ (see \cite{ReedSimonII} for details). In particular one has that
\begin{equation*}
\widehat{\psi}_N(f) \;=\; \int \md x \; \bar{f}(x) \, \widehat{\psi}_N(x)\,, \qquad \widehat{\psi}_N^*(f) \;=\; \int \md x \; f(x) \, \widehat{\psi}^*_N(x)\,.
\end{equation*}
Furthermore, the (anti)commutation relations \eqref{unrescaled commutation relations} imply that
\begin{equation} \label{anticommutation relations}
\qb{\widehat{\psi}_N(x), \widehat{\psi}_N^*(y)}_\mp \;=\; \frac{1}{N} \delta(x-y)\,,\qquad \qb{\widehat{\psi}_N^\#(x), \widehat{\psi}_N^\#(y)}_\mp \;=\; 0\,,
\end{equation}

\section{The Limit $\epsilon \to 0$ in Lemma \ref{lemma: Schwinger-Dyson expansion}} \label{section: removal of cutoff}
What remains is the justification of the equality in \eqref{final form of Schwinger-Dyson expansion} for $\epsilon = 0$. Our strategy is to show that both sides of \eqref{1/N expansion} with $\epsilon > 0$ converge strongly to the same expression with $\epsilon = 0$.

We first show the strong convergence of $G^{(k,l), \epsilon}_t(a^{(p)})$. Let $\Phi^{(n)} \in \cal{H}^{(n)}_\pm$ and consider
\begin{equation*}
\normb{(W^\epsilon_{ij,s} - W_{ij,s})\Phi^{(n)}} \;=\; \normb{I_{\{\abs{W_{ij}} > \epsilon^{-1}\}} W_{ij} e^{-is H_0}\Phi^{(n)}} \;\leq\; \normb{W_{ij} e^{-is H_0}\Phi^{(n)}}\,.
\end{equation*}
Since the right side is in $\LLL^1([0,t])$, we may use dominated convergence to conclude that
\begin{equation*}
\lim_{\epsilon \to 0} \int_0^t \md s \; \normb{(W^\epsilon_{ij,s} - W_{ij,s})\Phi^{(n)}} \;=\; 0\,.
\end{equation*}
Now
\begin{align*}
&\int_0^t \md s \int_0^t \md s' \; \normb{W^\epsilon_{ij, s} W^\epsilon_{i'j', s'} \Phi^{(n)} - W_{ij, s} W_{i'j', s'} \Phi^{(n)}}
\\
&\quad \leq\; \int_0^t \md s \int_0^t \md s' \; \normb{W^\epsilon_{ij, s} W^\epsilon_{i'j', s'} \Phi^{(n)} - W^\epsilon_{ij, s} W_{i'j', s'} \Phi^{(n)}}
\\
&\qquad
{}+{} 
\int_0^t \md s \int_0^t \md s' \; \normb{W^\epsilon_{ij, s} W_{i'j', s'} \Phi^{(n)} - W_{ij, s} W_{i'j', s'} \Phi^{(n)}}\,.
\end{align*}
The first term is bounded by
\begin{align*}
\pbb{\frac{\pi \kappa^2 t}{2}}^{1/2} \int_0^t \md s' \; \normb{W^\epsilon_{i'j', s'} \Phi^{(n)} - W_{i'j', s'} \Phi^{(n)}} \;\to\; 0\,,\qquad \epsilon \to 0\,.
\end{align*}
The integrand of the second term is bounded by $2 \normb{W_{ij, s} W_{i'j', s'} \Phi^{(n)}} \in \LLL^1([0,t]^2)$, so that dominated convergence implies that the second term vanishes in the limit $\epsilon \to 0$. A straightforward generalization of this argument shows that
\begin{equation*}
G^{(k,l), \epsilon}_t(a^{(p)}) \, \Phi^{(p+k-l)} \;\to\; G^{(k,l)}_t(a^{(p)}) \, \Phi^{(p+k-l)}\,,
\end{equation*}
as claimed.
Since the series \eqref{final form of Schwinger-Dyson expansion} converges uniformly in $\epsilon$, we find that
\begin{equation*}
\sum_{k = 0}^\infty \sum_{l = 0}^k \frac{1}{N^l} \, \widehat{\A}_N \pb{G^{(k,l),\epsilon}_t(a^{(p)})} \, \Phi^{(n)} \;\to \;
\sum_{k = 0}^\infty \sum_{l = 0}^k \frac{1}{N^l} \, \widehat{\A}_N \pb{G^{(k,l)}_t(a^{(p)})} \, \Phi^{(n)}\,,
\end{equation*}
as $\epsilon \to 0$.

Next, we show that $\me^{-\mi t H^\epsilon_N} \Phi^{(n)} \to \me^{-\mi t H_N} \Phi^{(n)}$. This follows from strong resolvent convergence of $H_N^\epsilon$ to $H_N$ as $\epsilon \to 0$ by Trotter's theorem \cite{ReedSimonI}. Let $W^\epsilon \deq \sum_{i<j} W^\epsilon_{ij}$, and consider 
\begin{align*}
N \, \normb{(H_N^\epsilon - \mi)^{-1} \Phi^{(n)} - (H_N - \mi)^{-1} \Phi^{(n)}} &\;=\; \normb{(H_N^\epsilon - \mi)^{-1} (W - W^\epsilon) (H_N - \mi)^{-1} \Phi^{(n)}} 
\\
&\;\leq\; \normb{(W - W^\epsilon) (H_N - \mi)^{-1} \Phi^{(n)}} \,.
\end{align*}
Clearly $\Psi^{(n)} \deq (H_N - \mi)^{-1} \Phi^{(n)}$ is in the domain of $H_N$. By the Kato-Rellich theorem \cite{ReedSimonII}, $\Psi^{(n)}$ is in the domain of $W_{ij}$ for all $i,j$. Therefore,
\begin{equation*}
\normb{(W_{ij} - W^\epsilon_{ij}) (H_N - \mi)^{-1} \Phi^{(n)}} \;=\; \normb{I_{\{\abs{W_{ij}} > \epsilon^{-1}\}} W_{ij} \Psi^{(n)}} \;\to\; 0
\end{equation*}
as $\epsilon \to 0$. Therefore
\begin{equation*}
\me^{\mi t H^\epsilon_N} \, \widehat{A}_N(a^{(p)}) \, \me^{- \mi t H^\epsilon_N} \Phi^{(n)} \;\to\; 
\me^{\mi t H_N} \, \widehat{A}_N(a^{(p)}) \, \me^{- \mi t H_N} \Phi^{(n)}
\end{equation*}
as $\epsilon \to 0$, and the proof is complete.

\end{document}